\documentclass{amsart}
\usepackage{amsaddr}
\usepackage[utf8]{inputenc}
\inputencoding{utf8}
\usepackage{booktabs}
\usepackage{amsmath}
\usepackage{amsthm}
\usepackage{amsfonts}
\usepackage{amssymb}
\usepackage{bbm}
\usepackage{float}
\usepackage{graphicx}
\usepackage{tikz}
\usepackage{natbib}
\usepackage{enumitem}
\usepackage{tikz-cd}
\usepackage{subcaption}
\usetikzlibrary{positioning}
\usepackage[ruled]{algorithm2e}
\usepackage[noend]{algpseudocode}
\usepackage[text={440pt,575pt},headheight=9pt,centering]{geometry}
\usepackage{mathtools}

\usepackage{hyperref}
\hypersetup{
	colorlinks=true,
	linkcolor=blue,
	filecolor=magenta,
	urlcolor=cyan,
	citecolor=blue
}
\usepackage[capitalize]{cleveref}

\allowdisplaybreaks

\newcommand{\x}{\mathbf{x}}
\newcommand{\y}{\mathbf{y}}

\newcommand{\p}{\mathbf{p}}
\newcommand{\w}{\mathbf{w}}
\newcommand{\1}{\mathbf{1}}
\newcommand{\0}{\mathbf{0}}
\newcommand{\A}{A}
\newcommand{\B}{B}
\newcommand{\X}{\mathbf{X}}
\newcommand{\Y}{\mathbf{Y}}
\newcommand{\W}{\mathbf{W}}

\newcommand{\RR}{\mathbb{R}}
\renewcommand{\S}{\mathcal{S}}

\newcommand{\EE}{\mathbb{E}}

\newcommand{\Det}{\text{Det}}

\newcommand{\diag}{\text{diag}}
\newcommand{\MTPtwo}{$ \text{MTP}_2 $ }

\newcommand{\EMTPtwo}{$ \text{EMTP}_2 $ }

\newcommand{\indep}{\perp \!\!\! \perp}

\DeclareMathOperator*{\tr}{\operatorname{tr}}

\DeclareMathOperator*{\argmin}{\operatorname{arg min}}

\DeclareMathOperator*{\sign}{\operatorname{sign}}
\DeclareMathOperator*{\rank}{\operatorname{rank}}

\newtheorem{thm}{Theorem}[section]

\newtheorem{cor}[thm]{Corollary}

\newtheorem{ex}{Example}

\newtheorem{rmk}{Remark}

\newcommand{\HR}{H\"usler--Reiss}

\begin{document}
\title{Latent Gaussian and Hüsler--Reiss Graphical Models with Golazo Penalty} 
\author{Ignacio Echave-Sustaeta Rodríguez and Frank R\"ottger}
\email{i.echave.sustaeta.rodriguez@tue.nl}
\email{f.rottger@tue.nl}
\date{February 4, 2025}
\address{Department of Mathematics and
	 Computer Science, Eindhoven University of Technology, The Netherlands}

\begin{abstract}
The existence of latent variables in practical problems is common, for example when some variables are difficult or expensive to measure, or simply unknown. 
When latent variables are unaccounted for, structure learning for Gaussian graphical models can be blurred by additional correlation between the observed variables that is incurred by the latent variables.
A standard approach for this problem is a latent version of the graphical lasso that splits the inverse covariance matrix into a sparse and a low-rank part that are penalized separately.
This approach has recently been extended successfully to \HR{} graphical models, which can be considered as an analogue of Gaussian graphical models in extreme value statistics. 
In this paper we propose a generalization of structure learning for Gaussian and \HR{} graphical models via the flexible Golazo penalty. 
This allows us to introduce latent versions of for example the adaptive lasso, positive dependence constraints or predetermined sparsity patterns, and combinations of those.
We develop algorithms for both latent graphical models with the Golazo penalty and demonstrate them on simulated and real data.
\end{abstract}
\maketitle

\section{Introduction}

In many inference problems it is common to implicitly assume that all  variables of interest are being observed and measured. 
This is however often not the case, for various reasons. 
For example, it is possible that there exist unknown factors that influence the observed variables. 
Alternatively, there may be variables which are too expensive or difficult to measure.
When our interest is in structure learning for Gaussian graphical models, in particular in high-dimensional settings, a common approach is covariance estimation with the graphical lasso \citep{MB2006,YL2007} which can recover the zero pattern of the inverse covariance matrix $K=\Sigma^{-1}$.
In the presence of latent variables however, a potentially sparse structure might be inaccessible.
Let $O$ denote the indices of the observed and $H$ the indices of the latent (or hidden) variables of some Gaussian random vector $\X$.
The inverse covariance matrix of $\X_O$ is the Schur complement
$$(\Sigma_{OO})^{-1} = K_{OO} - K_{OH}(K_{HH})^{-1}K_{HO}.$$
Here, even when the complete model is sparse, the subtrahend can blur the sparsity pattern in $K_{OO}$.
In other words, the latent variables incur correlations in the observed system which can render attempts to estimate directly the dependence structure of the system unsuccessful. 

For this setting \citet{Chandrasekaran} proposed to model the inverse observed covariance matrix as the difference of a sparse matrix $A=K_{OO}$ and a low-rank matrix $B=K_{OH}(K_{HH})^{-1}K_{HO}$.
They penalize sparsity in $A$ ($\ell_1$ norm) and rank in $B$ (the nuclear norm is the trace for symmetric PSD matrices), resulting in the following optimization problem
\begin{align}
	\big(\widehat{\A},\widehat{\B}\big) &= \underset{\A,\B}{\argmin}-\ell(\A-\B;S_{OO}) + \lambda_n(\gamma \lVert \A \rVert_1 + \tr(\B)), \label{eq:opt_prob_intro}
\end{align}
where $\ell$ is the Gaussian log-likelihood, $S_{OO}$ is the observed sample covariance, $A$ is required to be positive definite and $B$ to be positive semidefinite, and $\lambda_n$ and $\gamma$ are non-negative scalars.

In many applications, for example in finance or climate science,
extreme events and their dependence are of highest relevance. For instance, we
might be interested in modeling extremal dependence in a financial crisis or a climate disaster.
The recent introduction of conditional independence and graphical models in extremes from threshold exceedances \citep{engelkehitz} allows for sparse models that can capture complex extremal dependence structures.
Within these models, the parametric family of \HR{} distributions is particularly convenient for inference, as it permits a parametric encoding of extremal conditional independence \citep{HES2022}. 
Such models can be parameterized by a signed graph Laplacian matrix $\Theta$, where a zero entry imposes conditional independence.
Similar to latent Gaussian graphical models, the presence of unobserved variables leads to a \HR{} parameter matrix that is a Schur complement
\[\Theta_{OO} - \Theta_{OH}(\Theta_{HH})^{-1}\Theta_{HO}.\]
This observation allowed \cite{engelke2024extremal} to extend the approach \eqref{eq:opt_prob_intro} of \citet{Chandrasekaran} to \HR{} graphical models.
Note that similar ideas have been explored in the context of Laplacian-constrained Gaussian graphical models \citep{LI202367}.

In structure learning for multivariate Gaussians some alternatives to the $\ell_1$-penalty as in the graphical lasso have been proposed in the literature, for example the adaptive lasso \citep{FFW2009} or positive dependence \citep{LUZ2019}.
Recently, \citet{piotrsteffen} introduced the Golazo penalty as a flexible generalization of many penalties.
The Golazo penalty includes not only the adaptive lasso and positive dependence, but also allows for graphical model constraints or asymmetric penalties, and combinations of those.

In this paper we propose to modify the approaches of \citet{Chandrasekaran} and \citet{engelke2024extremal} using the Golazo penalty to allow more flexible structure learning in latent Gaussian and \HR{} graphical models.
This yields two related convex optimization problems, which we tackle with an alternating direction method of multipliers (ADMM) algorithm \citep{Chang2020}. For the Laplacian-constrained version
of our code we modify the algorithm from \cite{LI202367}.
We demonstrate the application of our method on simulated and real data. The real data for the Gaussian is obtained from \cite{Chang2020} but with the original source
being \cite{Hughes2000}. The real data for the extreme application is taken from \cite{engelke2024extremal}.

The code for this paper is publicly available on Github at \url{https://github.com/iechave-tue/golazo-latent-ggm-hr}.

\subsection{Notation}
Let $\S^d_>$ be the collection of all symmetric positive definite $d\times d$-matrices and $\S^d_\ge$ the cone of symmetric positive semidefinite $d\times d$-matrices.
We abbreviate $M_{\mathcal I, \mathcal J}$ to $M_{\mathcal{IJ}}$ for some matrix $M$ and index sets $\mathcal{I}, \mathcal{J}$.

\section{Preliminaries}\label{sec:prelim}
\subsection{Gaussian Graphical Model}
Let $\X\sim N(\mu,\Sigma)$ be a multivariate Gaussian with mean $\mu\in\RR^d$ and covariance $\Sigma\in \S^d_> $.
We call $K=\Sigma^{-1}$ the concentration matrix.
Let $G=(V,E)$ be a simple undirected graph with vertices $V=\{1,\ldots,d\}$ and edge set $E\subset V\times V$.
A Gaussian graphical model with respect to $G$ is the collection of all multivariate Gaussian distributions that satisfy
\begin{align}
	\forall\; ij\not\in E \Longrightarrow K_{ij}=0. \label{eq:GGM}
\end{align}
As $K_{ij}=0$ is equivalent to the conditional independence $X_i \indep X_j \mid X_{V\setminus\{i,j\}}$, the graph $G$ implies conditional independence constraints on $\X$.
As a slight abuse of notation, we will refer to any multivariate Gaussian $\X$ that satisfies \eqref{eq:GGM} with respect to some graph $G$ as a Gaussian graphical model.

\begin{ex}
Let $d = 4$ and let $G$ be the graph in \Cref{fig:example}. The graph $G$ implies zeros in $K$ as follows:
\[K=\begin{pmatrix}
	K_{11}&K_{12}&0&K_{14}\\
	K_{12}&K_{22}&K_{23}&0\\
	0&K_{23}&K_{33}&K_{34}\\
	K_{14}&0&K_{34}&K_{44}\\
\end{pmatrix}.\]
This is equivalent to conditional independence statements $X_1 \indep X_3 \mid X_{\{2,4\}}$ and $X_2 \indep X_4 \mid X_{\{1,3\}}$.
\begin{figure}[H]\centering
    \begin{tikzpicture}[roundnode/.style={circle, draw=black!60, very thick,minimum size=4mm,scale=1}]
                \node[roundnode]      (u1)                     {$1$};
                \node[roundnode]    (u2)    [above right=1cm of u1] {$2$};
                \node[roundnode]      (u4)       [below right = 1cm of u1] {$4$};
                \node[roundnode]      (u3)       [below right=1cm of u2] {$3$};
                
                \draw (u3) -- (u2);
                \draw (u3) -- (u4);
                \draw (u1) -- (u4);
                \draw (u1) -- (u2);
                
              \end{tikzpicture}
              \caption{Example of a Gaussian graphical model.}\label{fig:example}
    \end{figure}
\end{ex}

\subsection{Multivariate Gaussians with Hidden Variables}

Let $\X$ be a multivariate Gaussian. 
We assume to observe only the subvector of variables $\X_O$ with $O\subset [d]:=\{1,\ldots,d\}$, and consider the remaining variables $H$ as hidden, where $[d] = O \cup H$ and $O \cap H = \emptyset$. 
Given i.i.d.~(centered) observations $\{\x_O^1,\ldots,\x_O^n\}$ of $\X_O\sim N(\mathbf{0},\Sigma_{OO})$, we define the sample covariance matrix $S_{OO}=\frac{1}{n}\sum_{i=1}^{n}\x_O^i (\x_O^i)^T$.
The inverse covariance (concentration) matrix of $\X_O$ can be expressed in terms of the full concentration matrix $K$, such that 
\begin{align}
	(\Sigma_{OO})^{-1} = K_{OO} - K_{OH}(K_{HH})^{-1}K_{HO}.\label{eq:Observed_concentration}
\end{align}
Here, the right hand side is the Schur complement $K/K_{HH}$.

If the complete vector $\X$ satisfies certain constraints, e.g.~a sparsity pattern in $K$ as imposed by a Gaussian graphical model, the subset of observed variables $\X_O$ would by default not show the same constraints.
For example, the inverse covariance matrix $(\Sigma_{OO})^{-1}$ of the observed variables would typically be a dense matrix even when $K$ is sparse.
We illustrate this behavior with an example:

\begin{ex}
Let $\X$ be a 5-variate Gaussian vector that is Markov to the graph in Figure~\ref{fig:claw}. Therefore its concentration matrix $K$ satisfies
\begin{align*}
	K&=\begin{pmatrix}
		K_{11} & 0&0&0&K_{15}\\
		0 & K_{22}&0&0&K_{25}\\
		0 & 0&K_{33}&0&K_{35}\\
		0 & 0&0&K_{44}&K_{45}\\
		K_{15} & K_{25}&K_{35}&K_{45}&K_{55}\\
	\end{pmatrix},
\end{align*}
\begin{figure}[H]\centering
\begin{tikzpicture}[roundnode/.style={circle, draw=black!60, very thick,minimum size=4mm,scale=1}]
            \node[roundnode]      (h)                     {$H$};
            \node[roundnode]    (o1)    [below left=2cm and 2cm of h] {$O_1$};
            \node[roundnode]      (o2)       [below left=2cm and 0.3cm of h] {$O_2$};
            \node[roundnode]      (o3)       [below right=2cm and 0.3cm of h] {$O_3$};
            \node[roundnode]      (o4)       [below right=2cm and 2cm of h] {$O_4$};
            
            \draw[dotted] (h) -- (o1);
            \draw[dotted] (h) -- (o2);
            \draw[dotted] (h) -- (o3);
            \draw[dotted] (h) -- (o4);
            \draw (3,-1) -- (-3,-1);

          \end{tikzpicture}\hspace{5mm}
          \begin{tikzpicture}[roundnode/.style={circle, draw=black!60, very thick,minimum size=4mm,scale=1}]
          	\node[roundnode]      (u1)                     {$O_1$};
          	\node[roundnode]    (u4)    [above right=1cm of u1] {$O_4$};
          	\node[roundnode]      (u3)       [below right = 1cm of u1] {$O_3$};
          	\node[roundnode]      (u2)       [below right=1cm of u4] {$O_2$};
          	
          	\draw (u1) -- (u3);
          	\draw (u3) -- (u2);
          	\draw (u3) -- (u4);
          	\draw (u1) -- (u4);
          	\draw (u4) -- (u2);
          	\draw (u1) -- (u2);
          \end{tikzpicture}
          \caption{Graph with four observed variables and one hidden (left), and completely connected graph with four observed variables (right).}\label{fig:claw}
\end{figure}Here, we can see that the hidden variable is connected with all of the observed variables, while there are no edges between observed variables.
The observed subset of variables $\X_O$ has a dense concentration matrix 
\begin{align*}
	(\Sigma_{OO})^{-1}&=\begin{pmatrix}
		K_{11}-\frac{K_{15}^{2}}{K_{55}} & -\frac{K_{15}K_{25}}{K_{55}}&-\frac{K_{15}K_{35}}{K_{55}}&-\frac{K_{15}K_{45}}{K_{55}}\\
		-\frac{K_{15}K_{25}}{K_{55}} & K_{22}-\frac{K_{25}^{2}}{K_{55}}&-\frac{K_{25}K_{35}}{K_{55}}&-\frac{K_{25}K_{45}}{K_{55}}\\
		-\frac{K_{15}K_{35}}{K_{55}} & -\frac{K_{25}K_{35}}{K_{55}}&K_{33}-\frac{K_{35}^{2}}{K_{55}}&-\frac{K_{35}K_{45}}{K_{55}}\\
		-\frac{K_{15}K_{45}}{K_{55}} & -\frac{K_{25}K_{45}}{K_{55}}&-\frac{K_{35}K_{45}}{K_{55}}&K_{44}-\frac{K_{45}^{2}}{K_{55}}\\
	\end{pmatrix},
\end{align*}
such that the corresponding graphical model is completely connected. 
\end{ex}

In this setting, we would be interested in being able to estimate $K_{OO}$, since it gives us information about the sparsity of the full model, and also to estimate $K_{OH}(K_{HH})^{-1}K_{HO}$, since this matrix tells us information about the hidden variables. For instance, if $\vert H \vert$ is small, then it will have low rank, since its rank is bounded above by $\vert H \vert$. In particular, we can use an estimate of this matrix to estimate the number of hidden variables via the rank.

To tackle this problem, \cite{Chandrasekaran} proposed to penalize the two components $K_{OO}$ and $K_{OH}(K_{HH})^{-1}K_{HO}$ that form $(\Sigma_{OO})^{-1}$ separately.
To facilitate notation, we define $\A:=K_{OO}$ and $\B:=K_{OH}(K_{HH})^{-1}K_{HO}$.
Let $\ell(K;S) = \log\det(K) - \tr(KS)$ be the Gaussian log-likelihood for some concentration matrix $K$ and sample covariance $S$ as seen in \cite{Chandrasekaran}. They introduce the following optimization problem: 
\begin{align}
	\big(\widehat{\A},\widehat{\B}\big) &= \underset{\A\in \S^{d}_{>},\B\in \S^{d}_{\ge}}{\argmin}-\ell(\A-\B;S_{OO}) + \lambda_n(\gamma \lVert \A \rVert_1 + \tr(\B)). \label{eq:latent_opt}
\end{align}
Here, the $\ell_1$-norm penalty $\lVert \A \rVert_1 $ promotes the assumed sparsity, and the trace penalty term $\tr(\B)$ the low-rank constraint for $\B$, allowing us to try to estimate this hidden variable component without prior knowledge about it.

\citet[Theorem~4.1]{Chandrasekaran} provide a theoretical analysis of the convergence of the estimation above. Under a number of assumptions related with the tangent spaces of the sparse and low-rank matrices (please refer to \cite{Chandrasekaran} for details), the signs in $\A$ and the rank of $\B$ are estimated accurately with high probability. 

\begin{thm}\label{THM:CHANDRA} \citep[Theorem~4.1]{Chandrasekaran} Let $\A$ and $\B$ denote the ground-truth sparse and low-rank components. Let 
	\begin{equation*} g_{\gamma}(\A,\B) := \max\{\frac{1}{\gamma}\lVert \A \rVert_{\infty}, \lVert B \rVert_2\} \end{equation*} 
and given a matrix $M$ and its tangent space $T(M)$, let
\begin{equation*}
	\xi(T(M)) := \underset{N \in T(M),\lVert N \rVert_2 \leq 1}{\max} \lVert N \rVert_{\infty}.
\end{equation*}
Under the assumptions of \citet[Proposition~3.3 and Theorem~4.1]{Chandrasekaran}, we have that the probability of having simultaneously
\begin{itemize}
	\item $\sign(\A^*) = \sign(\widehat{\A}).$
	\item $\rank(\B^*) = \rank(\widehat{\B}).$
	\item $g_{\gamma}(\A^*-\widehat{\A},\B^*-\widehat{\B}) \lesssim \frac{1}{\xi(T(\B^*))}\sqrt{\frac{|O|}{n}}$
\end{itemize}
is at least $1-2\exp(-|O|)$.
\end{thm}

This result does not give us exactly consistency, since although we have error bounds depending on the sample size $n$, this does not happen with probability $1$ as $n$ goes to infinity. Instead this only happens with probability at least $1-2\exp(-|O|)$, which is however close to one with large enough $|O|$.

\subsection{Laplacian-constrained Gaussian graphical model}

Let $G=(V,E)$ be an undirected graph with weighted adjacency matrix $Q$. The signed Laplacian matrix of $G$ is a symmetric $d\times d$ matrix $\Theta $ with 
\[\Theta_{ij}=\begin{cases}
	-Q_{ij}, & i\neq j,\\
	\sum_{k=1}^{d}Q_{ik}, & i=j.\\
\end{cases}\]
Let $\mathcal{H}^{d-1}:=\{\x\in\RR^d: \x^T\1=\0\}$ be the hyperplane that is orthogonal to the all-ones vector.
A Laplacian-constrained Gaussian graphical model (LCGGM) is a random vector $\W\sim N(\mu,\Theta^+)$, where $\mu\in\mathcal{H}^{d-1}$ and $\Theta^{+}$ denotes the Moore--Penrose pseudoinverse of a positive semidefinite signed graph Laplacian matrix $\Theta\in \S^{d}_{\ge}$.
The random vector $\W$ has a probability density with respect to the Lebesgue measure on $\mathcal{H}^{d-1}$ that is
\[f_{\W}(\w)=\sqrt{(2\pi)^{-(d-1)}\Det(\Theta)}\exp\left(-\frac{1}{2}(\w-\mu)^T\Theta (\w-\mu)\right),\] 
where $\Det$ denotes the pseudodeterminant, i.e.~the product of all nonzero eigenvalues~\citep{NEURIPS2020_4ef42b32}.
This is an exponential family with natural parameter $Q$ and sufficient statistic $T(\W)$ with $T(\W)_{ij}=-\frac{1}{2}(W_i-W_j-(\mu_i-\mu_j))^2$ \citep{RS2022}.
The mean parameter is $\EE(T(\W))=-\frac{1}{2}\Gamma$ where $\Gamma$ is a variogram matrix with $\Gamma_{ij}=\operatorname{Var}(W_i-W_j)$.
Mean parameter and natural parameter are linked via the Fiedler--Bapat identity
\begin{align}
\begin{pmatrix}
	-\frac{1}{2}\Gamma &\1\\
	\1^T&0\\
\end{pmatrix}^{-1}=\begin{pmatrix}
\Theta&\p\\
\p^T&\sigma^2\\
\end{pmatrix}, \label{eq:fiedler}
\end{align}
where $\p=\frac{\Gamma^{-1}\1}{\1^T\Gamma\1}$ and $\sigma^2=\frac{1}{2\1^T\Gamma^{-1}\1}$.
The matrix $\Gamma$ is conditionally negative definite, i.e.~$\Gamma\in \mathcal C^d = \{\Gamma \in [0,\infty)^{d\times d} : \Gamma = \Gamma^T, \diag(\Gamma) = \0, v^T\Gamma v < 0 \: \forall v \perp \1, v \neq \0\}$, which implies that $\Theta$ is positive semidefinite.

\subsection{Multivariate extremes and the Hüsler--Reiss distribution}\label{subsec:extremes}

Let $\X$ be a $d$-dimensional random vector.
When interest is in extremal dependence, one can assume that all margins of $\X$ are standardized.
Here, as in \cite{HES2022},
we consider exponential margins for $\X$, i.e.~for all $i\in [d]$
we have that $\mathbb P(X_i \leq x) = 1-\exp(-x)$ for nonnegative values of $x$.

The limit of threshold exceedances
\begin{align}
	\mathbb P(\Y \leq \y) &= \lim_{u\to\infty} \mathbb P(\X - u\1 \leq\y \mid \X \nleq u\1),\label{eq:MPD_limit}
\end{align}
if it exists, gives rise to a multivariate Pareto distribution with support $\mathcal L = \{x \in \mathbb R^d : x \nleq \0\}$~\citep{RT2006}.
Here, one says that $\X$ is in the domain of attraction of $\Y$.
The distribution of the random vector $\Y$ can be expressed as
\begin{align*}	
	\mathbb P(\Y \leq \y) &= \frac{\Lambda^c(\y \land \0)-\Lambda^c(\y)}{\Lambda^c(\0)},
\end{align*}
where $\Lambda^c(\y) := \Lambda([-\infty,\infty)^d \setminus [-\infty,\y])$ and $\Lambda$
is a measure on $[-\infty,\infty)^d \setminus \{-\boldsymbol{\infty}\}$, usually called the exponent measure \citep{engelke2024graphicalmodelsmultivariateextremes}.
This measure is finite on sets bounded away from $-\boldsymbol{\infty}$, which ensures
that the previous expression for the distribution of $\Y$ is well-defined.

If we assume that the exponent measure $\Lambda$ is absolutely continuous with
respect to the Lebesgue measure in $d$ dimensions, we can consider the so-called
exponent measure density $\lambda$, the Radon--Nikodym derivative of the exponent measure.

Note that the restriction of $\lambda$ to $\mathcal L$ is proportional to the density $f$ of
$\Y$, so it is possible to write the density as $f(\y) = \lambda(y) / \Lambda^c(\0)$.

\citet{engelkehitz} introduced an extremal notion of conditional independence for multivariate Pareto distributions via factorization of the exponent measure density.
Let $\lambda_A(\y_A)=\int_{\RR^{d-|A|}}\lambda(\y)d\y_{[d]\setminus A}$ be the marginal exponent measure density for some $A\subseteq[d]$.
It holds that $\lambda_A$ is the exponent measure density of the threshold exceedance limit of $\X_A$ in \eqref{eq:MPD_limit}, compare \citet{HES2022}.
For disjoint subsets $A,B,C\subseteq[d]$ we say that $\Y_A$ is conditionally independent of $\Y_B$ given $\Y_C$ (in short, $\Y_A\perp_e\Y_B|\Y_C$) when
\[\lambda_{A\cup B\cup C}(\y_{A\cup B\cup C})\lambda_{C}(\y_{C})=\lambda_{A\cup C}(\y_{A\cup C})\lambda_{B\cup C}(\y_{ B\cup C})\]
for all $\y \in \mathcal{L}$. 
For some undirected graph $G=([d],E)$, we then call a multivariate Pareto vector $\Y$ an extremal graphical model with respect to $G$ when
\[(i,j)\not\in E\; \Longrightarrow\; \Y_i\perp_e\Y_j|\Y_{[d]\setminus ij}.\]

In this paper, we will focus on the parametric family of H\"usler--Reiss distributions.
This parametric family of multivariate Pareto distributions is parameterized by a variogram matrix $\Gamma\in\mathcal{C}^{d}$.
Let $\W$ be an LCGGM with $\mu=(I-\frac{1}{d}\1\1^T)(-\frac{1}{2}\Gamma)\1$ and precision matrix $\Theta$. Then, the exponent measure density of a \HR{} has a representation
\begin{align}
	\lambda(\y)=c_{\Gamma}\exp(-\frac{1}{d}\y^T\1)f_\W(\y),\label{eq:HR_density}
\end{align}
where $c_\Gamma>0$ is a normalizing constant.
The marginal exponent measure density $\lambda_{A}(\y_A)$ is of the shape \eqref{eq:HR_density} with variogram $\Gamma_{A,A}\in \mathcal{C}^{|A|}$.
It holds that 
\[Y_i\perp_e Y_j|\Y_{[d]\setminus\{i,j\}}\Longleftrightarrow \Theta_{ij}=0.\]
Thus, imposing sparsity in $\Theta$ imposes sparsity in the corresponding \HR{} graphical model.

\subsection{Latent \HR{} graphical models}

\cite{engelke2024extremal} introduced a general latent \HR{} graphical model as follows.
Let some random vector $\X=(\X_O,\X_H)$ be in the domain of attraction of a \HR{} vector $\Y$ (compare \eqref{eq:MPD_limit}) with parameters
\begin{align*}
\Gamma&=\begin{pmatrix}
	\Gamma_{OO} &\Gamma_{OH}\\
	\Gamma_{HO} & \Gamma_{HH}\\
\end{pmatrix},&&\Theta=\begin{pmatrix}
		\Theta_{OO} &\Theta_{OH}\\
		\Theta_{HO} & \Theta_{HH}\\
	\end{pmatrix}.
\end{align*} 
Then, the random vector $\X_O$ is in the domain of attraction of a \HR{} distribution with variogram $\Gamma_{OO}$.
The precision matrix corresponding to $\Gamma_{OO}$ can be obtained for example via the Fiedler--Bapat identity \eqref{eq:fiedler} and calculates as the Schur complement
$$\widetilde{\Theta} = \Theta_{OO} - \Theta_{OH}(\Theta_{HH})^{-1}\Theta_{HO},$$
compare also \cite{engelke2024extremal}.
Now, if $\X$ is in the domain of attraction of a sparse \HR{} graphical model, the underlying parameter matrix $\Theta$ is sparse. 
However, if we only observe $\X_O$ for some $O\subset [d]$, the sparsity pattern in $\Theta_{OO}$ will be masked by the low-rank component $\Theta_{OH}(\Theta_{HH})^{-1}\Theta_{HO}$.
To tackle this problem one can employ similar strategies as for latent Gaussian graphical models. 
We decompose $\widetilde{\Theta} := A-B$ into a sparse
part $A$ (ideally, a matrix close to $\Theta_{OO})$ and a low-rank part $B$
(ideally close to $\Theta_{OH}(\Theta_{HH})^{-1}\Theta_{HO}$).

A natural next step would be a penalized maximum likelihood approach similar to \eqref{eq:latent_opt}, but the shape of the \HR{} log-likelihood complicates this approach.
As an alternative, \cite{engelke2024extremal} propose a surrogate maximum likelihood method based on previous work of \cite{HES2022,totalpositivity}, where the \HR{} log-likelihood gets replaced by a mean-zero LCGGM log-likelihood.

Assuming i.i.d.~observations of $\X_O$, the empirical variogram $\overline{\Gamma}_{OO}$ of \cite{engelkevolgushev} is a consistent estimator of $\Gamma_{OO}$, see Section~\ref{sec:application} for more details on its construction.
Using the LCGGM log-likelihood with $\overline{\Gamma}_{OO}$ as the summary statistic, this gives rise to the optimization problem
\[\argmin_{A,B \mid (A-B)\in\S^{d}_{\ge},B\in\S^{d}_{\ge}} -\log\Det(A-B)-\frac{1}{2}\tr((A-B)\overline{\Gamma}_{OO})+\lambda_n(\gamma||A||_{1}+\tr(B)), \]
for positive scalars $\lambda_n,\gamma$ and under the constraint that $\tilde{\Theta}=A-B$ is a positive semidefinite signed graph Laplacian.
The problem with this setting is that the pseudo-determinant is computationally inconvenient. As can
be seen in \cite{LI202367} (where they work in the context of LCGGMs,
suggesting the same approach but in the non-extreme context),
there are ways to rewrite this to get a better expression for computations.
We can write $\tilde{\Theta}=A-B = P\Xi P^T$, where $P \in \mathbb R^{d\times (d-1)}$
is the orthogonal complement of $\1$ and $\Xi \in \mathbb R^{(d-1)\times(d-1)}$ 
is non-singular. Now, we have that $\Det(A-B) = \det((A-B) + \1^T\1/d)=\det(\Xi)$ and $\tr((A-B)\overline{\Gamma}_{OO})=\tr(P\Xi P^T\overline{\Gamma}_{OO})=\tr(\Xi P^T\overline{\Gamma}_{OO}P)$, so we can 
optimize in terms of the matrix $\Xi$. In the case of \citet{engelke2024extremal}, they write
the problem in a different way, but it seems that when solving it with a convex solver
(in their provided code), such a structure is useful. The optimization problem is now as follows:
\begin{gather}
	(\hat{\Xi}_n,\hat{A}_n,\hat{B}_n) = \underset{\Xi\in\S^{d}_{\ge};A;B\in\S^{d}_{\ge}}{\argmin}-\log\det(\Xi) -\frac{1}{2} \tr(\Xi P^T\overline{\Gamma}_{OO}P) + \lambda_n(\gamma\lVert A \rVert_{1} + \tr(B))\label{eq:latentHRopt} \\
	\textrm{s.t. } P\Xi P^T = A - B.\nonumber
	\end{gather}
We choose to write the optimization problem in such a way since it helps us to decompose it in subproblems for solving
it using an ADMM-based algorithm.

\subsection{Golazo Constraints}\label{sec:Golazo}
\cite{piotrsteffen} introduce the Golazo penalty function: $$\lVert K \rVert_{LU} = \sum\limits_{i,j}\max\{L_{ij}K_{ij},U_{ij}K_{ij}\}.$$
Here, $L,U$ are matrices with entries in $\mathbb R \cup \{\infty, -\infty\}$ such that $L_{ij} \leq 0 \leq U_{ij}$ for all $i,j \in [d]$.
For a given sample covariance $S$, adding the Golazo penalty to the negative Gaussian log-likelihood gives rise to a flexible penalized estimation procedure
\begin{align*}
	\widehat{K}&=\argmin_{K\succeq 0} -\ell(K;S) + \lVert K \rVert_{LU},
\end{align*}
that generalizes the standard $\ell_1$-penalty as in the graphical lasso.
The same idea can be applied to the surrogate \HR{} maximum likelihood problem with an LCGGM log-likelihood with signed Laplacian parameter $\Theta$ and the empirical variogram $\overline{\Gamma}$ as summary statistic.
Among the possible constraints that can be enforced with the Golazo penalty are the following:
\begin{itemize}
	\item \textbf{Asymmetric adaptive graphical lasso}: Let $L_{ij} = l_{ij} <0$ and $U_{ij} = u_{ij}>0$ for all $i, j$. 
	With this, it is possible to penalize differently positive and negative entries. When $L_{ij}=-U_{ij}$ for all $i, j$ we are in the adaptive graphical lasso framework, see \cite{FFW2009} for details.
	If $-l_{ij}=u_{ij}=\lambda_n$ for all $i, j$ for some scalar $\lambda_n$, we have the usual symmetric graphical lasso.
	\item \textbf{Positive lasso}: If we only want to penalize positive entries, we set $L_{ij} = 0$ and $U_{ij} = \lambda_n >0$. 
	\item \textbf{\MTPtwo distributions}: A multivariate Gaussian is multivariate totally positive of order two ($\text{MTP}_2$) if and only if $K_{ij}\le 0$ for all $i\neq j$ \citep{LUZ2019}. Setting $L_{ij} = 0$ and $U_{ij} = \infty$ for all $i\neq j$ yields the Gaussian MLE under \MTPtwo when $\lVert K \rVert_{LU}$ penalizes the log-likelihood.
	For a \HR{} distribution, the constraint $\Theta_{ij}\le 0$ for all $i\neq j$ is equivalent to a notion of extremal \MTPtwo ($\text{EMTP}_2$) \citep{totalpositivity}.
	Setting $L_{ij} = 0$ and $U_{ij} = \infty$ for all $i\neq j$ yields the \HR{} surrogate MLE under \EMTPtwo when $\lVert \Theta \rVert_{LU}$ penalizes the surrogate log-likelihood.
	\item \textbf{Positivity and sparsity}: It is possible to constrain for (extremal) \MTPtwo and additionally enforce sparsity by setting $L_{ij} = -\lambda_n <0$ and $U_{ij} = \infty$ for all $i\neq j$. 
	\item \textbf{Gaussian/ \HR{} graphical models}: If by assumption / domain knowledge we wish to set the entry $K_{ij}$ or $\Theta_{ij}$ to $0$, it is possible to enforce this by setting $-L_{ij} = U_{ij} = \infty$, under the convention that $0 \cdot \pm\infty  = 0$.
\end{itemize}

\section{Learning latent Gaussian and \HR{} graphical models via Golazo constraints}
\subsection{Gaussian setting}
The main idea of this section is to introduce more flexible latent variable modeling for multivariate Gaussians. 
For this we propose to substitute the $\ell_1$-penalty in the latent optimization problem \eqref{eq:latent_opt} with the Golazo penalty.
This allows to incorporate custom constraints for the dependence structure of $A=K_{OO}$, see Section~\ref{sec:Golazo} for a list of examples.
We thus propose the following optimization problem:
\begin{equation}
	\big(\widehat{\A},\widehat{\B}\big) = \underset{\A\in \S^{d}_{>},\B\in \S^{d}_{\ge}}{\argmin}-\ell(\A-\B;S_{OO}) + \lVert \A \rVert_{LU} + \lambda_n \tr(\B). \label{eq:latent_opt_golazo}
\end{equation}
Note that here the regularization constants can be absorbed into the $L,U$ parameters of the Golazo penalty, so we don't include them explicitly. 
The log-likelihood $\ell(K;S)$ is a strictly concave function in $K$. The Golazo penalty is convex \citep{piotrsteffen}. Thus the optimization problem $\eqref{eq:latent_opt_golazo}$ is convex.

\citet{Chandrasekaran} provide an asymptotic result (see \Cref{THM:CHANDRA}) for the latent Gaussian graphical lasso.
 The following corollary of \Cref{THM:CHANDRA} extends their result to certain asymmetric Golazo constraints in which we change the $\ell_1$ penalty weight in the off-diagonal entries. 
 We believe that a similar result should hold for arbitrary Golazo constraints.
\begin{cor}\label{cor:asymptotic}
	Let $K$ be the true inverse covariance matrix and define $\A,\B$ as before. Let all the assumptions of \Cref{THM:CHANDRA} be satisfied, including 
	the choice of $\lambda_n$ and $\gamma$. Then, define the Golazo parameters $L,U$ such that
	\begin{itemize}
		\item if $A^*_{ij} > 0$, choose $L_{ij} \in [-\infty, -\lambda_n\gamma]$ and let $U_{ij}=\lambda_n\gamma$,
		\item if $A^*_{ij} < 0$, let $L_{ij}=-\lambda_n\gamma$ and choose $U_{ij} \in [\lambda_n\gamma, \infty].$
		\item if $A^*_{ij} = 0$, choose $L_{ij} \in [-\infty, -\lambda_n\gamma]$ and $U_{ij} \in [\lambda_n\gamma, \infty].$ 
	\end{itemize}
	In this case we recover the correct sign pattern of $A^*$ and rank of $B^*$ with probability greater than $1-2\exp(-|O|)$.
\end{cor}
\begin{proof}
	The original statement (when $U_{ij} = -L_{ij} = \lambda_n\gamma$) tells us that with probability larger than 
	$1-2\exp(-|O|)$, the sign of the estimate $\widehat{A}$ is equal to that of $A^*$, and the rank of $\widehat{B}$
	is the same as that of $B^*$.  This means that with that probability,
	the optimal point of the problem in \Cref{eq:latent_opt_golazo} has the correct signs and rank. 
	
	In general, if we add a larger positive
	penalty to any non-optimal points, the optimal point will stay the same. Here, if $A^*_{ij}>0$ is positive, we can increase
	the penalty on the negative values by making $L_{ij}$ smaller. Similarly, if $A^*_{ij} < 0$, we can increase the
	penalty on positive points by increasing $U_{ij}$. Finally, if $A^*_{ij} = 0$, then we can increase both penalties
	simultaneously while maintaining the same optimal point. This proves that the statement about sign and rank is still
	satisfied.
\end{proof}
	Corollary~\ref{cor:asymptotic} implies that any sign constraints (such as enforcing positivity in an entry, or
	enforcing sparsity) can be added without losing guarantees if such an assumption is accurate in the specific practical setting.
	A positive entry in the matrix is enforced by fixing the corresponding entry of $L$ to $-\infty$, a negative entry
	is enforced by fixing the corresponding entry of $U$ to $\infty$, and a zero is enforced by doing both simultaneously.
	Thus, Corollary~\ref{cor:asymptotic} extends the result of \cite{Chandrasekaran} to any setting where the ground truth satisfies such constraints.

\subsection{Learning Laplacian-constrained Gaussian Graphical Models under Golazo Constraints}
In this section we propose to generalize the approach of \citet{engelke2024extremal} via the Golazo penalty.
Assume the setting of $\eqref{eq:latentHRopt}$ for a given empirical variogram $\overline{\Gamma}_{OO}$.
As for latent Gaussian graphical models, the Golazo penalty allows more flexible constraints on $A=\Theta_{OO}$ (see Section~\ref{sec:Golazo}) than the original $\ell_1$ penalty.
This gives rise to the optimization problem
\begin{gather*}
(\hat{\Xi}_n,\hat{A}_n,\hat{B}_n) = \underset{\Xi\in\S^{d}_{\ge};A;B\in\S^{d}_{\ge}}{\argmin}-\log\det(\Xi) -\frac{1}{2} \tr(\Xi P^T\overline{\Gamma}_{OO}P) + \lVert A \rVert_{LU} + \lambda_n\tr(B) \\
\textrm{s.t. } P\Xi P^T = A - B,
\end{gather*}
where $\lambda_n$ is a positive scalar.
It is possible to write this optimization problem in similar but slightly different ways, the one we show here will be useful for the two-block ADMM algorithm in the following section that we will use to solve the problem, in which this structure in terms of 3 blocks of variables is natural. 

\begin{rmk}
	As we mentioned above, \citet{LI202367} work with the same optimization problem as \citet{engelke2024extremal},
	only that instead of using $-\overline{\Gamma}_{OO}/2$ as their data input, they use a sample covariance matrix $S_{OO}$.
	Furthermore, they only consider Laplacian matrices, i.e.~only positive edge weights.
	Thus, our generalized approach can also be applied in their setting. In particular, the positivity constraints of Laplacian matrices can be captured easily by Golazo.
	However, note that although \HR{} models and LCGGM models are profoundly related as illustrated by the shape of the \HR{} exponent measure density \eqref{eq:HR_density}, the marginal $\W_O$ will not be parameterized by the signed Laplacian $\tilde{\Theta}$ as $\W_O$ is not degenerate for any proper subset $O\subset [d]$.
\end{rmk}

\section{ADMM Algorithm}\label{sec:admm}
\subsection{Gaussian setting}
To tackle the convex optimization problem \eqref{eq:latent_opt_golazo} it is possible to use a general convex solver. For this paper we will employ a multi-block ADMM algorithm that is often used for solving similar problems in the machine learning context, given that this methods can give better time performance than a general convex solver by taking advantage of separable problems in terms of the blocks of variables. Here, we are adapting the algorithm studied in \cite{Chang2020}, which is a good reference for the details on the general idea of the algorithm. 
We rewrite \eqref{eq:latent_opt_golazo} in terms of three blocks of variables as follows:
\begin{gather}
	(\widehat{M},\widehat{\A},\widehat{\B}) = \underset{M,\A\in \S^{d}_{>},\B\in \S^{d}_{\ge}}{\argmin}-\ell(M;S_{OO}) + \lVert \A \rVert_{LU} + \lambda_n\tr(\B)\quad \textrm{  s.t. } M = \A - \B. \label{eq:opt_prob_ADMM}
\end{gather}
We define the augmented Lagrangian of the optimization problem $$\mathcal L_{\sigma}(M,\A,\B,\Lambda) := -\ell(M;S_{OO}) + \lVert A \rVert_{LU} + \lambda_n\tr(\B) -\langle \Lambda, M - \A + \B\rangle + \frac{\sigma}{2}\rVert M - \A + \B \rVert^2,$$
where $\Lambda \in \mathbb R^{d\times d}$ are the Lagrange multipliers. This algorithm used this augmented Lagrangian
since the additional penalty helps enforce the constraints between the blocks of variables. Here, $\sigma$ denotes the
hyperparameter that tunes how strongly the constraints between the blocks of variables are enforced.
The $k+1$ iteration of the algorithm will be as follows:
\begin{equation*}
\begin{cases}
M^{k+1} := \underset{M \in \mathbb R^{d\times d}}{\argmin} \: \mathcal L_{\sigma}(M,\A^k, \B^k, \Lambda^k)+\frac{\rho\sigma}{2}\lVert M-M^k\rVert^2,\\
\Lambda^{k+\frac{1}{2}}:= \Lambda^k - \alpha\sigma(M^{k+1} - \A^k + \B^k),\\
\A^{k+1} := \underset{\A \in \mathbb R^{d\times d}}{\argmin} \: \lVert \A \rVert_{LU} + \frac{\tau r_1}{2}\Big\lVert \A - \A^k + \frac{\Lambda^{k+\frac{1}{2}}}{\tau r_1}\Big\rVert^2, \\
\B^{k+1} := \underset{\B \in \mathbb \B^{d\times d}\; \B\succeq 0}{\argmin} \: \lambda_n \tr(\B) + \frac{\tau r_2}{2}\Big\lVert \B - \B^k + \frac{\Lambda^{k+\frac{1}{2}}}{\tau r_2}\Big\rVert^2, \\
\Lambda^{k+1} := \Lambda^{k+\frac{1}{2}} + \sigma(\A^{k+1} - \A^k) - \sigma(\B^{k+1} - \B^k).
\end{cases}
\end{equation*}
The Lagrange multiplier is updated two times in each iteration given the multi-block nature of the problem.
For details about the procedure see \cite{Bai2017}. As shown in \cite{Chang2020}, the conditions 
$\tau \in (\frac{2+\alpha}{2}, +\infty), \rho \in [0,+\infty), r_1 > \sigma, r_2 > \sigma$ are sufficient conditions 
for convergence. Here $\alpha$ is the step size of the half-update of the Lagrange multiplier. It is suggested by them to fix for practical reasons $\rho = 0$, 
$\tau = \varsigma\frac{2+\alpha}{2}$ and $r_1 = r_2 = \varsigma\sigma$, where $\varsigma > 1$. Here, $\rho$ is a parameter
than can help speed up convergence of the method but we do not worry about this in our paper.

The three subproblems that we have after the considerations about the parameters have simple closed form solutions, which we briefly summarize here. Firstly, the subproblem for $M^{k+1}$ has a first order condition 
$$S_{OO} - M^{-1} + \sigma \big(M - \A^k + \B^k - \frac{\Lambda^k}{\sigma}\big) + \rho\sigma(M-M^k) = 0.$$
By multiplying by $M$, this is converted into a quadratic equation on $M$: $$(\rho+1)\sigma M^2 + \big(S_{OO} + \sigma(\B^k - \A^k) - \Lambda^k - \rho\sigma M^k\big)M - I = 0.$$
If we consider the eigendecomposition 
$C \diag(\mathbf{v}) C^T = S_{OO} + \sigma(B^k - A^k) - \Lambda^k - \rho\sigma M^k$
and define a new vector of eigenvalues $\x$ such that $$x_i := \frac{-v_i + \sqrt{v_i^2 + 4(\rho + 1)\sigma}}{2(\rho+1)\sigma},$$ then the closed form solution to the problem is $M^{k+1} = C \diag(\x) C^T.$

For the second subproblem, let $\0$ denote the zero matrix and let $\max$ denote here the entry-wise maximum. Then, the solution is $$A^{k+1} = \min\Bigg\{A^k - \frac{\Lambda^{k+\frac{1}{2}}+L}{\tau r_1},\0\Bigg\} + \max\Bigg\{A^k - \frac{\Lambda^{k+\frac{1}{2}} - U}{\tau r_1},\0\Bigg\}.$$

Finally, the third subproblem also has a simple closed form solution. Consider the eigendecomposition $D \diag(\beta) D^T = B^k + \frac{\Lambda^{k+\frac{1}{2}}-\lambda_n I}{\tau r_2}.$ Then, the closed form solution is given by $B^{k+1} = D \diag(\max(\beta,\0))D^T,$ where again the $\max$ is taken entry-wise.

Therefore, it is straightforward to solve this problem iteratively. Let $N$ denote the maximum number of iterations that we allow in a practical setting, and let $\epsilon_1, \epsilon_2 \in \mathbb R_{\geq 0}$ be parameters such that the algorithm stops if we have that both of the following conditions are satisfied:

\begin{align*}
\text{RelChg} :&= \max\Bigg\{ \frac{\lVert M^{k+1} - M^k \rVert_F}{1+\lVert M^k \rVert_F}, \frac{\lVert A^{k+1} - A^k \rVert_F}{1+\lVert A^k \rVert_F}, \frac{\lVert B^{k+1} - B^k \rVert_F}{1+\lVert B^k \rVert_F} \Bigg\} < \epsilon_1 \\
\text{IER} :&= \lVert M^k - A^k + B^k \rVert_F < \epsilon_2.
\end{align*}

The algorithm stops only after the maximum number of iterations are performed or when the previous criterion is satisfied. We show pseudocode for the algorithm in \Cref{ALG:GGM}.

\begin{algorithm} 
\caption{Multi-block ADMM for GGM estimation}\label{ALG:GGM}
 \DontPrintSemicolon
 \LinesNumbered
\KwIn{$S_{OO}, L, U, P, \{\sigma, \alpha, r_1, r_2, \tau, \lambda_n, \rho\}$, \{$\epsilon_1, \epsilon_2$\}, $N$, $k=0$}
\KwOut{$\hat{M}_n,\hat{A}_n,\hat{B}_n$}
Starting point: $M^0 \leftarrow I, A^0 \leftarrow I, B^0 \leftarrow \0$\;
\While{$k<N$ and (RelChg $\geq \epsilon_1$ or IER $\geq \epsilon_2$)}{Compute eigendecomposition $C \diag(\alpha) C^T$ of $S_{OO} + \sigma(B^k - A^k) - \Lambda^k - \rho\sigma M^k$\;
$x_i \leftarrow \frac{-\alpha_i + \sqrt{\alpha_i^2 + 4(\rho + 1)\sigma}}{2(\rho+1)\sigma}$\;
$M^{k+1} \leftarrow C \diag(\x) C^T$ \;
$\Lambda^{k+\frac{1}{2}} = \Lambda^k - \alpha\sigma(M^{k+1}- A^k + B^k)$\;
$A^{k+1} \leftarrow \min\Bigg\{A^k - \frac{\Lambda^{k+\frac{1}{2}}+L}{\tau r_1},\0\Bigg\} + \max\Bigg\{A^k - \frac{\Lambda^{k+\frac{1}{2}} - U}{\tau r_1},\0\Bigg\}$\;
Compute eigendecomposition $D \diag(\beta) D^T$ of $B^k + \frac{\Lambda^{k+\frac{1}{2}}-\lambda_n I}{\tau r_2}$\;
$B^{k+1} \leftarrow D \diag(\max(\beta,\0))D^T$\;
$\Lambda^{k+1} \leftarrow \Lambda^{k+\frac{1}{2}} + \sigma(A^{k+1} - A^k) - \sigma(B^{k+1} - B^k)$\;
$k=k+1$\;
}
\Return $\hat{M}_n \leftarrow M^k,\hat{A}_n \leftarrow A^k,\hat{B}_n \leftarrow B^k$\;
\end{algorithm}

\subsection{\HR{} setting}
In the \HR{} setting the algorithm is analogous to the previous one, with a slight 
modification in the first update. This happens because we are now optimizing over $\Xi$, although 
the subproblem is analogue to the one that we have for $M$. It can be seen in \citet{LI202367}
and in Algorithm~\ref{ALG:LAP} how to take into account this small difference, where first we solve
the subproblem for $\Xi$ (thus obtaining $\Xi^{k+1}$) and then we obtain the new
estimate for $\Theta$ by simply setting $\Theta^{k+1} := P\Xi^{k+1}P^T$.

In the pseudocode in 
\Cref{ALG:LAP}, the data input $S$ is meant to be $S_{OO}$ if we would like to estimate
a signed LCGGM, and $-\overline{\Gamma}_{OO}/2$ if we would like to 
estimate an HR extremal graphical model. 
We can use the same 
algorithm since the only difference in the optimization problem is this one. Such 
an algorithm was already considered by \citet{LI202367}. In their case they solve the 
problem for latent Laplacian-constrained Gaussian Graphical Models with the $\ell_1$ penalty,
so our approach here generalizes the constraints and the possibility to include extremes as in
\cite{engelke2024extremal}.

In addition to this, they solve a slightly more constrained problem,
where the signed Laplacians that appear are instead Laplacians (that is, the off-diagonal
entries can only be smaller or equal to 0). With our approach, we generalize this to allow any positive semidefinite
signed Laplacian.

\begin{algorithm} 
\caption{Multi-block ADMM for Laplacian graphical model estimation}\label{ALG:LAP}
 \DontPrintSemicolon
 \LinesNumbered
\KwIn{$S, L, U, P, \{\sigma, \alpha, r_1, r_2, \tau, \lambda_n, \rho\}$, \{$\epsilon_1, \epsilon_2$\}, $K$, $k=0$}
\KwOut{$\hat{\Theta}_n,\hat{A}_n,\hat{B}_n$}
Starting point: $\Theta^0 \leftarrow I, \Xi^0 = P^TP, A^0 \leftarrow I, B^0 \leftarrow \0$\;
\While{$k<K$ and (RelChg $\geq \epsilon_1$ or IER $\geq \epsilon_2$)}{Compute eigendecomposition $C \diag(\alpha) C^T$ of $P^T S P + \sigma P^T(A^k - B^k)P - P^T\Lambda^kP - \rho\sigma \Xi^k$\;
$x_i \leftarrow \frac{-\alpha_i + \sqrt{\alpha_i^2 + 4(\rho + 1)\sigma}}{2(\rho+1)\sigma}$\;
$\Xi^{k+1} \leftarrow C \diag(\x) C^T$ \;
$\Theta^{k+1} = P\Xi^{k+1}P^T$\;
$\Lambda^{k+\frac{1}{2}} \leftarrow \Lambda^k - \alpha\sigma(\Theta^{k+1}- A^k + B^k)$\;
$A^{k+1} \leftarrow \min\Bigg\{A^k - \frac{\Lambda^{k+\frac{1}{2}}+L}{\tau r_1},\0\Bigg\} + \max\Bigg\{A^k - \frac{\Lambda^{k+\frac{1}{2}} - U}{\tau r_1},\0\Bigg\}$\;
Compute eigendecomposition $D \diag(\beta) D^T$ of $B^k + \frac{\Lambda^{k+\frac{1}{2}}-\lambda_n I}{\tau r_2}$\;
$B^{k+1} \leftarrow D \diag(\max(\beta,\0))D^T$\;
$\Lambda^{k+1} \leftarrow \Lambda^{k+\frac{1}{2}} + \sigma(A^{k+1} - A^k) - \sigma(B^{k+1} - B^k)$\;
$k=k+1$\;
}
\Return $\hat{\Theta}_n \leftarrow \Theta^k,\hat{A}_n \leftarrow A^k,\hat{B}_n \leftarrow B^k$\;
\end{algorithm}

\section{Application}\label{sec:application}
Note that during our experiments we will fix the values of the ADMM parameters following the practical choices made in the
paper of \cite{Chang2020}, that is, we do not tune these values for speed, we only pick values that guarantee convergence of the method.
Additionally, we will focus on parameters $L,U$ for the Golazo penalty that do not penalize the diagonal, so for our practical
experiments $\diag(L) = \diag(U) = \0$.
\subsection{Simulated Data}
\subsubsection{Gaussian setting}
Taking inspiration from \cite{engelke2024extremal}, we consider a graph $G=(V,E)$ consisting of two disconnected (except for edges going through the hidden variable) 
cycles with $25$ observed nodes each, and one hidden variable. We set $K_{ii}=5$ for all $i\in V = \{1,\ldots,51\}$ 
and $K_{ij}=-2$ for all $1\le i, j \le p=50$ with $ij\in E$, and $K_{ij}=0$ otherwise. The hidden variable is connected
to all of the observed variables, with $K_{ih} = K_{hi} = 5/p$ for all $i\neq h= 51$. 

In this study we compare the standard $\ell_1$-penalty with a positive dependence constraint.
To showcase the flexibility of the Golazo approach, we further include two modified versions of these penalties that incorporate partial graphical model constraints (i.e.~partial sparsity in $K$).
To simplify notation, let us call $O_1 = \{1,\ldots,25\}, O_2 = \{26, \ldots 50\}$, $H = \{51\}$, where
$O_1$ denotes the indices of the nodes of the first cycle, $O_2$ the nodes of the second cycle and $H$ the hidden variable.
The constraints that we are going to test are the following:
\begin{enumerate}
	\item $L_{ij} = -\lambda_n\gamma$ and $U_{ij} = \lambda_n\gamma$ for all $i\neq j$, that is, the standard $\ell_1$-penalty.
	\item $L_{ij} = -\lambda_n\gamma$ and $U_{ij} = \lambda_n\gamma$ for all $i\neq j$ where $i,j$ are both either in $O_1$ or $O_2$.
	For $i,j$ where each node is in a different subcycle, $L_{ij} = -\infty$ and $U_{ij} = \infty$, that is, we assume that $O_1$ or $O_2$ are not connected by an edge.
	\item $L_{ij} = 0$ and $U_{ij} = \infty$ for all $i\neq j$, that is, the \MTPtwo constraint.
	\item $L_{ij} = 0$ and $U_{ij} = \infty$ for all $i\neq j$ where $i,j$ are both either in $O_1$ or $O_2$.
	For $i,j$ where each node is in a different subcycle, $L_{ij} = -\infty$ and $U_{ij} = \infty$, that is, the \MTPtwo constraint with the additional assumption that $O_1$ or $O_2$ are not connected by an edge.
\end{enumerate}

\begin{figure}[H]\centering
	\begin{tikzpicture}[roundnode/.style={circle, draw=black!60, very thick,minimum size=4mm,scale=1}]
		\node[roundnode]    (o1)     {$O_1$};
		\node[roundnode]      (o2)       [right=2cm of o1] {$O_2$};
		\node[roundnode]      (o4)       [below=2cm of o1] {$O_4$};
		\node[roundnode]      (o3)       [right=2cm of o4] {$O_3$};
		\node[roundnode]      (h)       [below right=1cm and 1cm of o2] {$H$};
		\node[roundnode]    (o5)     		[above right=1cm and 1cm of h]				{$O_5$};
		\node[roundnode]      (o6)       [right=2cm of o5] {$O_6$};
		\node[roundnode]      (o8)       [below=2cm of o5] {$O_8$};
		\node[roundnode]      (o7)       [right=2cm of o8] {$O_7$};

		\draw (o1) -- (o2);
		\draw (o2) -- (o3);
		\draw (o3) -- (o4);
		\draw (o4) -- (o1);
		\draw (o5) -- (o6);
		\draw (o6) -- (o7);
		\draw (o7) -- (o8);
		\draw (o8) -- (o5);
		\draw[dotted] (h) -- (o1);
		\draw[dotted] (h) -- (o2);
		\draw[dotted] (h) -- (o3);
		\draw[dotted] (h) -- (o4);
		\draw[dotted] (h) -- (o5);
		\draw[dotted] (h) -- (o6);
		\draw[dotted] (h) -- (o7);
		\draw[dotted] (h) -- (o8);
	\end{tikzpicture}
	\caption{Two disconnected $4$-cycles with one hidden variable.}\label{fig:simulated}
\end{figure}

\begin{figure}[H]
	\centering
	\includegraphics[scale=0.6]{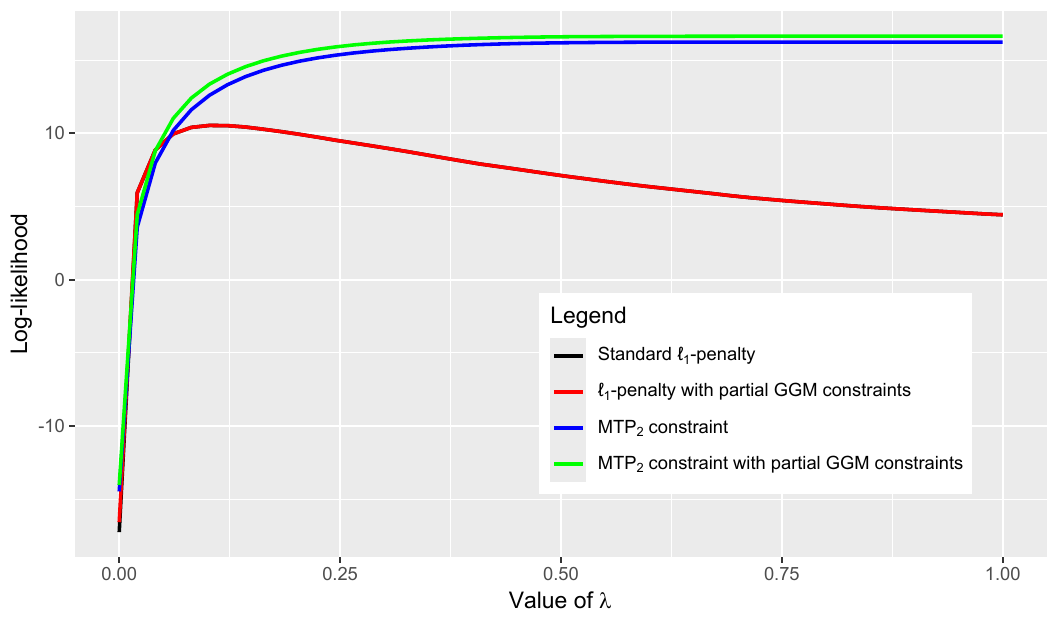}
	\caption{Results for the two cycles (red and black line become equal).}\label{fig:twocycles}
\end{figure}

We generate two samples of size $n=100$ in $N=20$ different trials. We train the model using the first sample and then 
evaluate the Gaussian log-likelihood on the second one. This could also be done using the ground truth covariance.
We fix $\gamma = 0.5$ for the constraints 1~and 2, after testing various values and noticing that the overall behavior is stable for a range of values of $\gamma$ (compare also the discussions in \cite{Chandrasekaran,CSPW2011,engelke2024extremal}).

Note that constraints 3~and 4~are independent of $\gamma$.
We select values for $\lambda_n$ from $10^{-8}$ to $1$, with $50$ values evaluated in total.
We perform the simulation, calculation and validation steps for each constraint and value of $\lambda_n$ and compute an average of the log-likelihoods over the different trials. 
Figure~\ref{fig:twocycles} visualizes the results of this study.
We observe that the \MTPtwo constraints provide a robust method that outperform the $\ell_1$-penalty. Furthermore, there is a small improvement when the partial graphical model constraints are added. 

\subsubsection{\HR{} setting}
In this simulation we will use some of the ideas from \citet{engelke2024extremal}. We consider the following graph structure.
The dependence graph between the observed variables is a cycle, we further connect each observed variable to one of the hidden variables.
The way to do so is to connect each hidden node $h \in H$ to all nodes $o \in O$ such that $o = h - (p+1) + \xi h$, for some positive integer $\xi$.
The weight between the
observed variables that are adjacent is $2$, and the weight between an observed variable and the hidden variable connected to
it is sampled uniformly in the interval $[50/\sqrt{p/h},75/\sqrt{p/h}]$. We will study three models generated using 
the function generate\_latent\_model\_cycle from \cite{engelke2024extremal}, with $n=10$, $N=10$, $p=30$ and $h=3,5,10$, to see the different behavior
depending on the number of hidden variables. We fix $\gamma = 0.25$ and $\lambda_n$ from $10^{-4}$ to $0.16$. The results can be
seen in \Cref{fig:hrlogli,fig:hredges,fig:hrranks}. The first column
shows the results for $h=3$, the second shows the results for $h=5$ and the third one for $h=10$.
We show the results for three different Golazo constraints:
\begin{enumerate}
	\item $L_{ij} = -\lambda_n\gamma$ and $U_{ij} = \lambda_n\gamma$ for all $i\neq j$, that is, the standard $\ell_1$-penalty.
	\item $L_{ij} = -\lambda_n\gamma$ and $U_{ij} = \infty$ for all $i\neq j$, that is, a modified $\ell_1$-penalty with the \EMTPtwo constraint.
	\item $L_{ij} = 0$ and $U_{ij} = \infty$ for all $i\neq j$, that is, the \EMTPtwo constraint.
\end{enumerate}

\begin{figure}[htbp]
	\centering
	\begin{minipage}{0.33\textwidth}
	  \centering
	  \includegraphics[width=\linewidth]{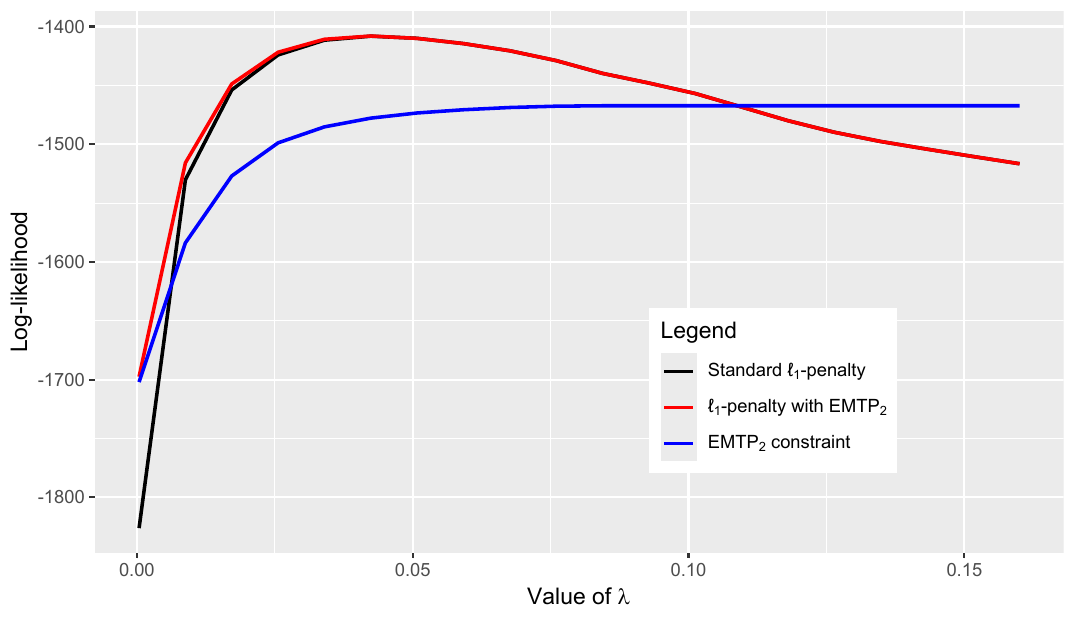}
	\end{minipage}\hfill
	\begin{minipage}{0.33\textwidth}
	  \centering
	  \includegraphics[width=\linewidth]{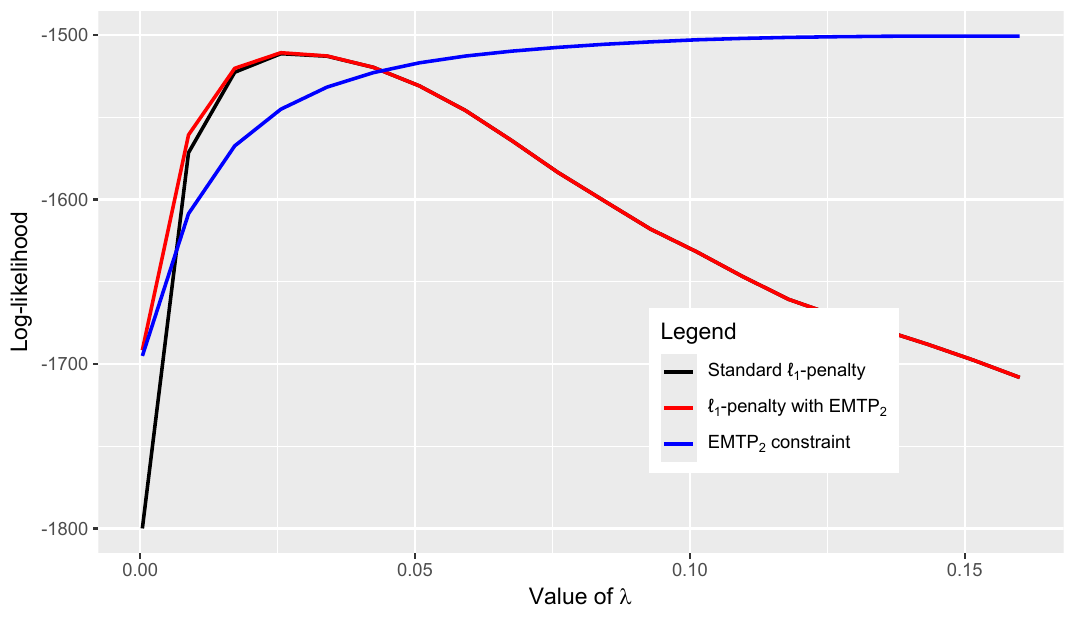}
	\end{minipage}\hfill
	\begin{minipage}{0.33\textwidth}
	  \centering
	  \includegraphics[width=\linewidth]{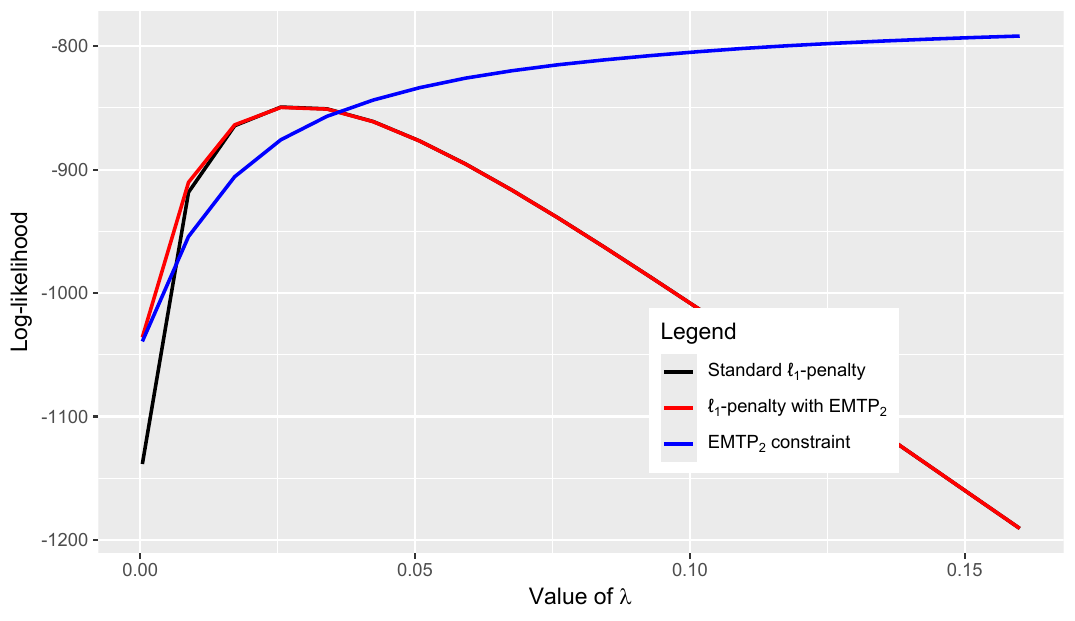}
	\end{minipage}
	\caption{Average validation Hüsler--Reiss log-likelihood}\label{fig:hrlogli}
  \end{figure}

\begin{figure}[htbp]
\centering
\begin{minipage}{0.33\textwidth}
	\centering
	\includegraphics[width=\linewidth]{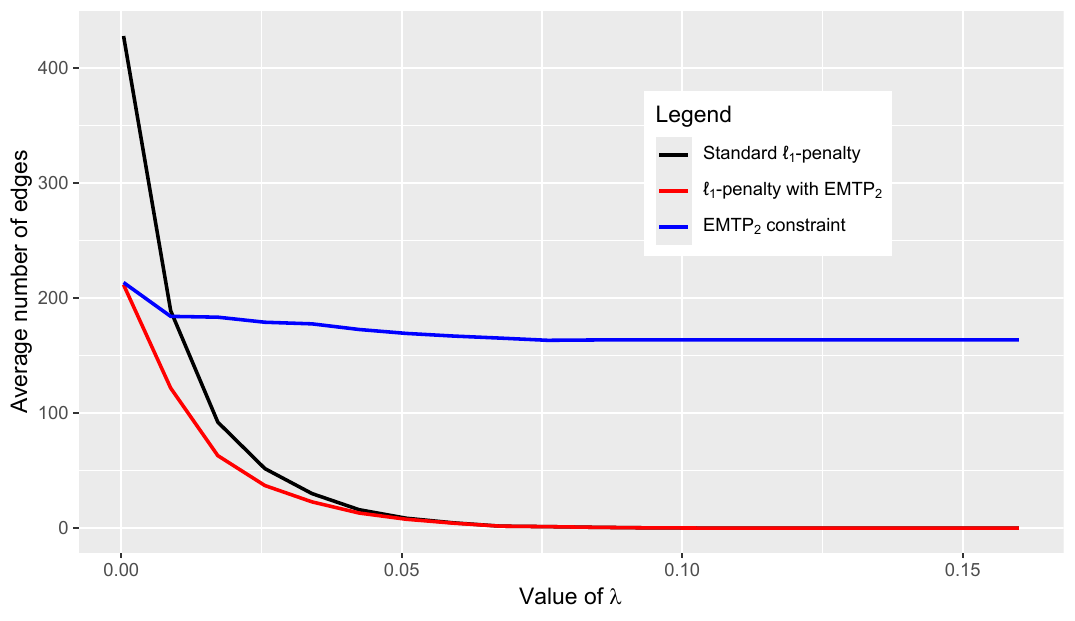}
\end{minipage}\hfill
\begin{minipage}{0.33\textwidth}
	\centering
	\includegraphics[width=\linewidth]{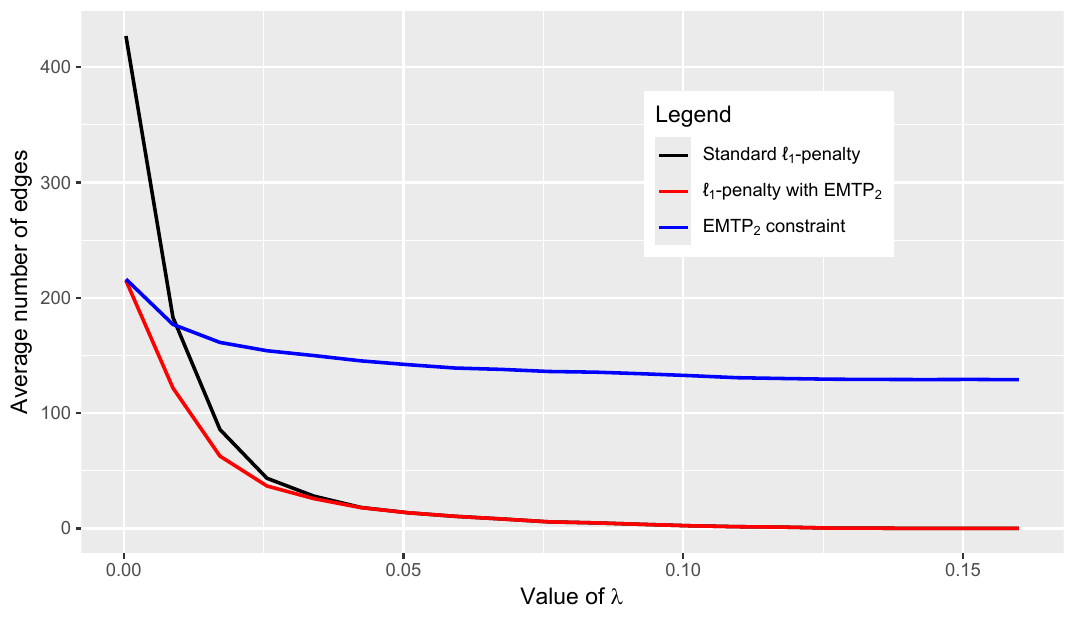}
\end{minipage}\hfill
\begin{minipage}{0.33\textwidth}
	\centering
	\includegraphics[width=\linewidth]{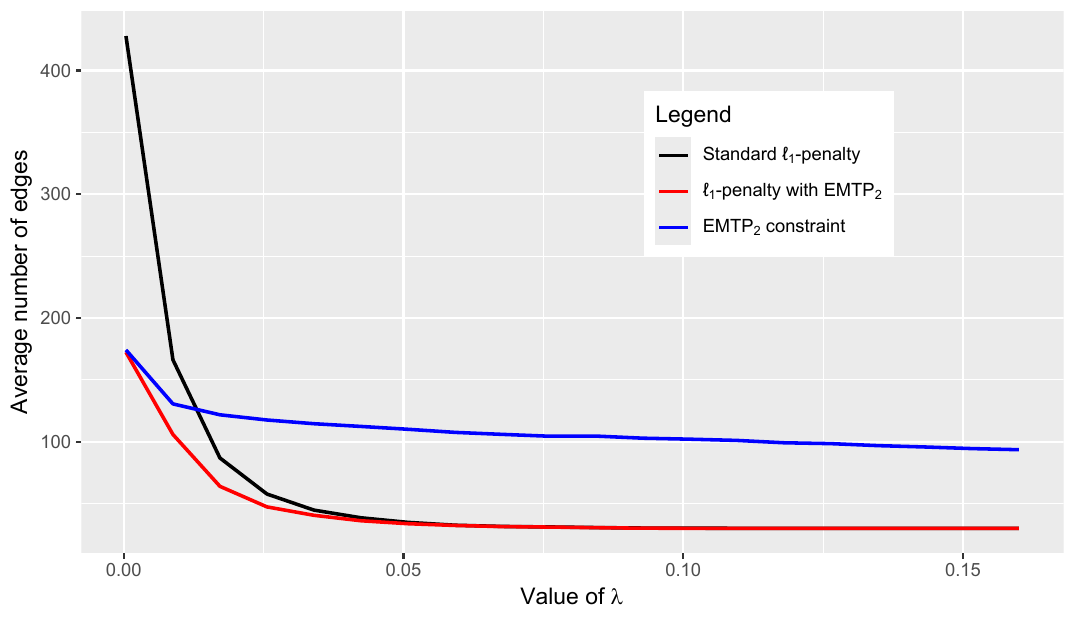}
\end{minipage}
\caption{Average estimated edges}\label{fig:hredges}
\end{figure}

\begin{figure}[htbp]
	\centering
	\begin{minipage}{0.33\textwidth}
	  \centering
	  \includegraphics[width=\linewidth]{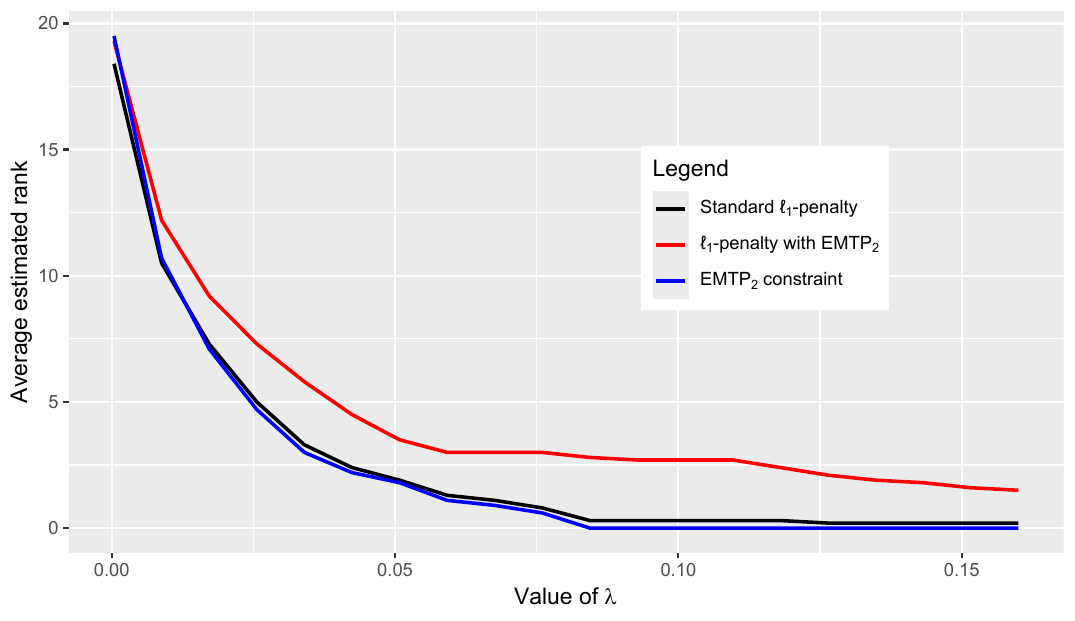}
	\end{minipage}\hfill
	\begin{minipage}{0.33\textwidth}
	  \centering
	  \includegraphics[width=\linewidth]{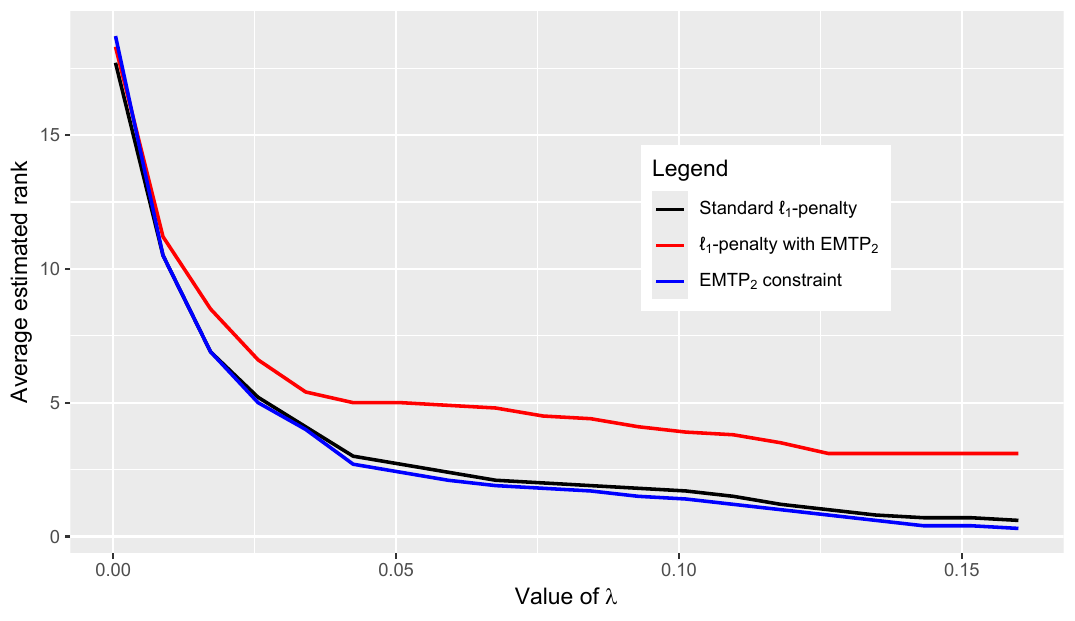}
	\end{minipage}\hfill
	\begin{minipage}{0.33\textwidth}
	  \centering
	  \includegraphics[width=\linewidth]{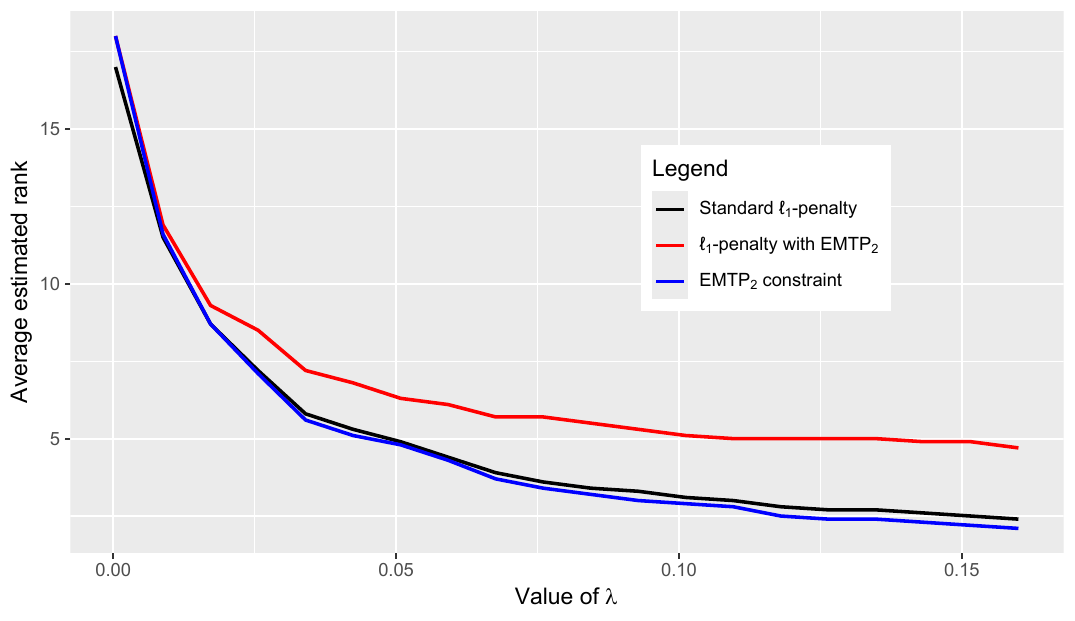}
	\end{minipage}
	\caption{Average estimated ranks}\label{fig:hrranks}
  \end{figure}

We see that the behavior in the different measurements depends of many factors. Regarding the validation log-likelihood,
we see that with $h=3$ the lasso-based methods perform best. However, the more hidden variables
we consider, the better are the results that the \EMTPtwo constraint obtains. Taking a look into the estimation of the number
of edges, we see however that the positivity constraint is not very sparse compared to the other ones,
which is one of the clear disadvantages when considering the tradeoff between model performance and computational
efficiency. Finally, when looking at the estimated ranks, we see that the estimated rank near the optimum
of the combination of lasso with \EMTPtwo seems to be quite close on average to the actual number of
hidden variables. This highlights the different behaviors observed between different Golazo constraints.

\subsection{Real-world Data}
\subsubsection{Standard Gaussian setting}
For this real-world data application we will use gene data from the Rosetta dataset (see \cite{Hughes2000} for the original source), which has $301$ samples
from $6316$ variables. We obtained the dataset from the code of \cite{Chang2020}. The way to process this data to 
obtain a sample covariance matrix (which is the data input to our algorithm) is described in \cite{MXZ2013}. 
Here, the idea is to compute the sample variances of each variable, and then pick the $p$ variables with the
largest sample variance, resulting in $p=25$ observed variables for the latent Gaussian graphical model.

During these experiments, we fix $\gamma = 0.1$, after testing various values and seeing that this one gave near optimal
result for the lasso-based methods. We select it in this way since the positivity-based methods optimal performance is not affected by this parameter.
Then we explore how the behavior of the estimates depend on the value
of $\lambda_n$ and the type of Golazo constraint selected. We select a large enough interval for $\lambda_n$ so that the general behavior
of each constraint can be appreciated. Here $\lambda_n$ takes values from $10^{-8}$ to $0.4$, with $30$ values evaluated in total.

We use $5$-fold cross-validation to evaluate how well each
of the methods generalizes better, and we will use as the score the log-likelihood with respect to the validation set.
We show the results for four different Golazo constraints:
\begin{enumerate}
	\item $L_{ij} = -\lambda_n\gamma$ and $U_{ij} = \lambda_n\gamma$ for all $i\neq j$, that is, the standard $\ell_1$-penalty.
	\item $L_{ij} = -\lambda_n\gamma$ and $U_{ij} = \infty$ for all $i\neq j$, that is, a modified $\ell_1$-penalty.
	\item $L_{ij} = 0$ and $U_{ij} = \infty$ for all $i\neq j$, that is, the \MTPtwo constraint.
	\item $L_{ij} = 0$ and $U_{ij} = \lambda_n\gamma$ for all $i\neq j$, that is, the positive lasso constraint.
\end{enumerate}
We can see in Figure~\ref{fig:gene} that the best overall validation log-likelihood occurs when using constraint 2, which shows that combining \MTPtwo and an $\ell_1$-penalty can yield improved performance over either of them. We see as in the simulation study that the \MTPtwo constraint
seems to be relatively robust with respect to the choice of $\lambda_n$ and performs comparably well, although not optimal in this case.
\begin{figure}
	\centering
	\includegraphics[scale=0.6]{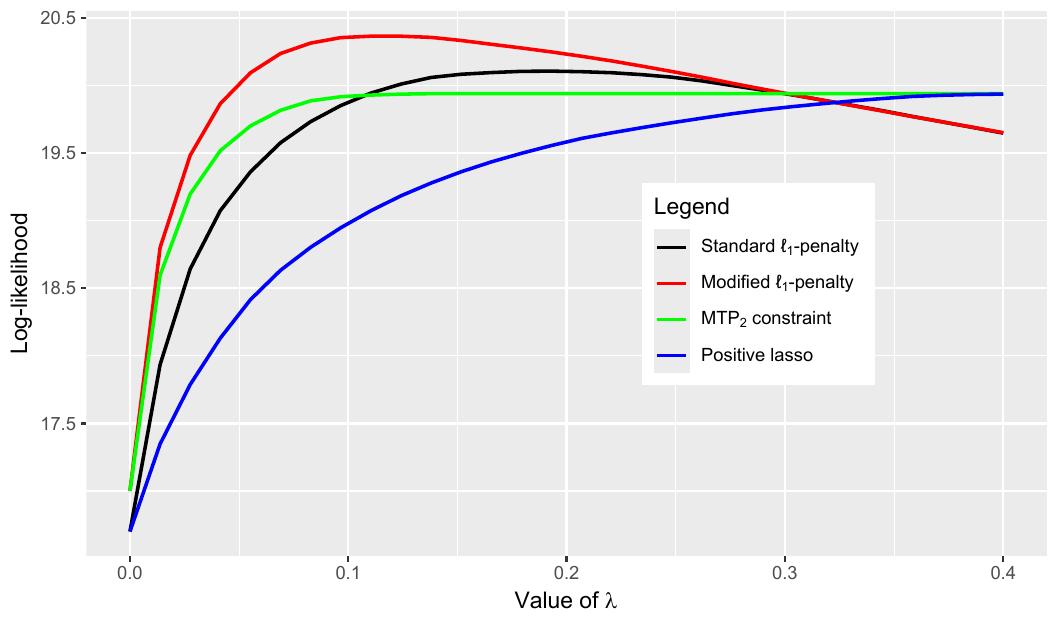}
	\caption{Results for the gene data (red and black line become equal).}\label{fig:gene}
\end{figure}

\subsubsection{\HR{} setting}
For our real-world data application in the extreme case, we choose to tackle the same problem as \citet{engelke2024extremal}.
We thank the authors for providing their code, which we reuse in our simulations and plots.
The dataset contains $n = 3603$ observations of total daily flight delays between 2005 and 2020 in $p = 29$ airports located in the south of the United States. This is an interesting problem to consider for our method, since extreme flight delays
can be caused by a wide variety of external factors. 
Furthermore, the setting is interesting for extreme value statistics as the impact of large flight delays is 
considerably larger than in the case of shorter flight delays. Secondly, we expect that the dependence structure
among the extremes should not be the same as for regular delays. 

The data input to our algorithm, as mentioned in \Cref{sec:admm}, is $-\overline{\Gamma}_{OO}/2$,
where $\overline{\Gamma}_{OO}$ is the empirical variogram of \citet{engelkevolgushev}.
In the approach of threshold exceedances as in \eqref{eq:MPD_limit}, a vector is considered extreme if its largest value exceeds a high threshold.
As the exceedance can be located in any entry of the observed vector, the approach of Engelke and Volgushev is to check this separately in every dimension.
Given observations $\X_1,\ldots,\X_n$, the $m$-th empirical variogram is defined as
$$\overline{\Gamma}^{(m)}_{ij} = \widehat{\textrm{Var}}(\log(1-\tilde{F_i}(X_{ti})) - \log(1-\tilde{F_j}(X_{tj})) : \tilde{F_m}(X_{tm}) \geq 1-k/n).$$
Here, $0\le k\le n$ is an integer that determines the effective sample size, $\widehat{\textrm{Var}}$ is the sample variance, and $\tilde{F_i}$ is the empirical distribution function for the dimension $i$.
Under the assumption that the data-generating process $\X$ is in the domain of attraction of a \HR{} vector $\Y$, the population version
\[\Gamma_{ij}^{(m)}=\operatorname{Var}(Y_i-Y_j|Y_m>1)\]
is constant with respect to $m$, that is~$\Gamma^{(1)}=\ldots\Gamma^{(m)}=\Gamma$.
Thus, averaging over all dimensions leads to the empirical variogram estimator
$$\overline{\Gamma} = \frac{1}{p}\sum_{m \in V}\overline{\Gamma}^{(m)}.$$
In this way, the effective sample size for estimating each of the $\overline{\Gamma}^{(m)}$ is
$k$, and the effective sample size for the full estimate depends on the dependence
structure in each case, since some data points might be extreme in many dimensions
(therefore being used multiple times, reducing the effective sample size), or the complete
opposite, where data points are only extreme in one dimension if any (causing an 
effective sample size close to $pk$).

To determine the effective sample size we typically select a value $a \in [0,1)$,
which marks the quantile threshold for our data to be considered extreme. Then the effective sample size is $k = (1-a)n$. 
In \citet{engelke2024extremal}, the authors discuss the cases $a\in\{0.85,0.9,0.95\}$, and find similar results in all settings.
As the largest value leads to the most extreme data set with the smallest effective sample size, we pick $a=0.95$ for the remainder of this section.

The parameter $\lambda_n$ ranges from $10^{-10}$ to $0.4$, with $10$ different values being evaluated.
As in the Gaussian case above, we fix the value of $\gamma = 0.25$. Please note that in the current version of the code we use a different
notation for some of the parameters. 

We show the results for three different Golazo constraints:
\begin{enumerate}
	\item $L_{ij} = -\lambda_n\gamma$ and $U_{ij} = \lambda_n\gamma$ for all $i\neq j$, that is, the standard $\ell_1$-penalty.
	\item $L_{ij} = -\lambda_n\gamma$ and $U_{ij} = \infty$ for all $i\neq j$, that is, a modified $\ell_1$-penalty with the \EMTPtwo constraint.
	\item $L_{ij} = 0$ and $U_{ij} = \infty$ for all $i\neq j$, that is, the \EMTPtwo constraint.
\end{enumerate}
For this experiment, we do $5$-fold cross-validation, and we report the average \HR{} log-likelihoods (\Cref{fig:flightloglikeli}) and number
of edges (\Cref{fig:flightedges}) for each of the Golazo constraints along all the evaluated values of $\lambda_n$.

\begin{figure}
	\centering
	\includegraphics[scale=0.6]{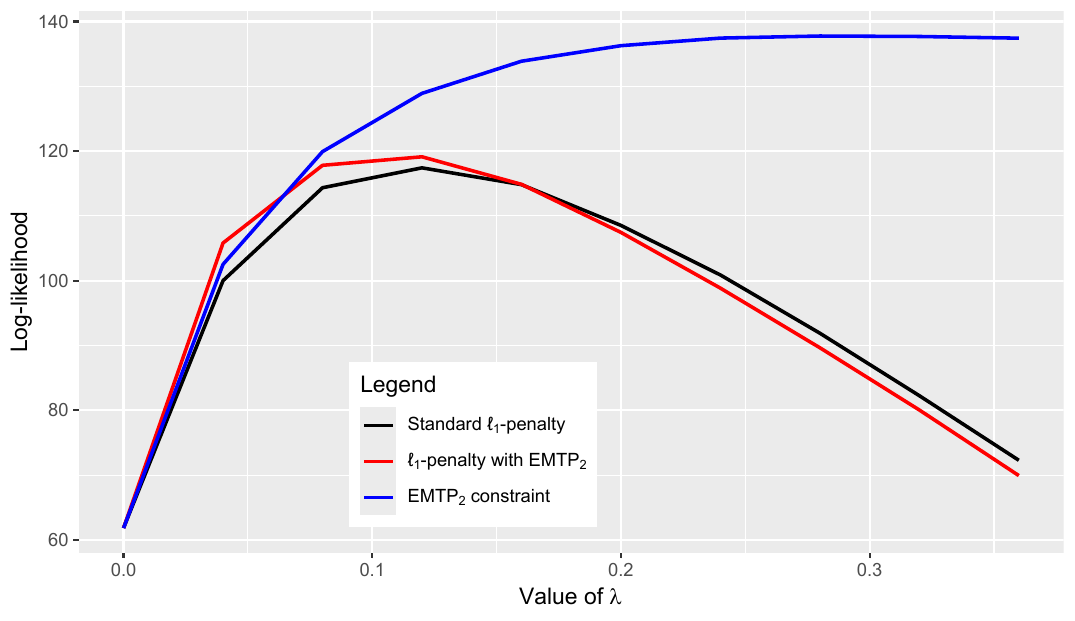}
	\caption{Results in average cross-validation \HR{} log-likelihood for the flights data.}\label{fig:flightloglikeli}
\end{figure}

\begin{figure}
	\centering
	\includegraphics[scale=0.6]{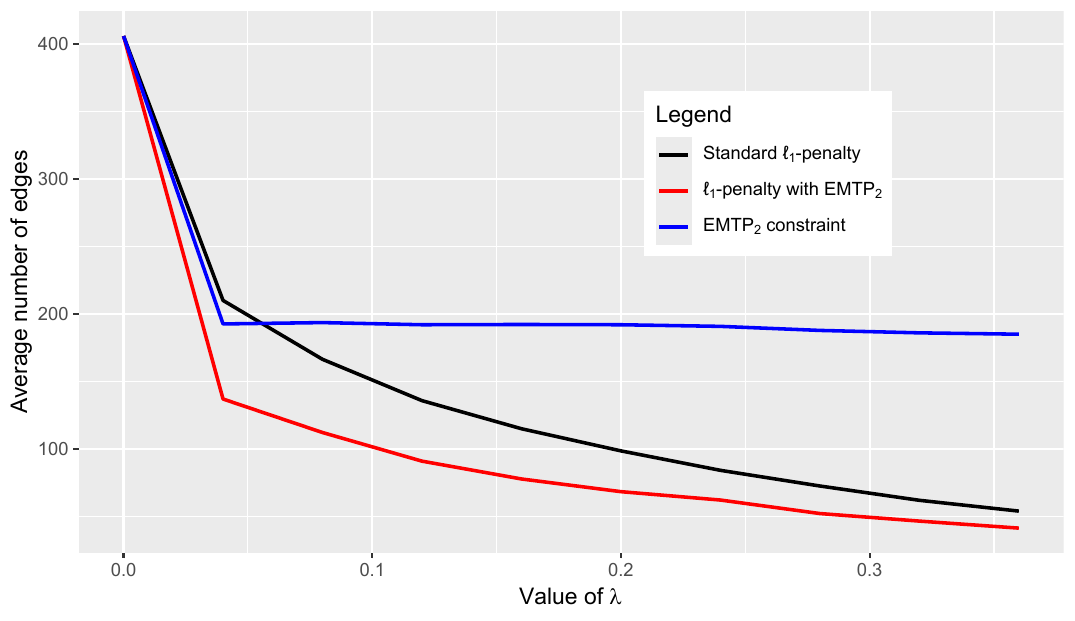}
	\caption{Results in average number of edges in cross-validation for the flights data}\label{fig:flightedges}
\end{figure}

As we can see, the value of the validation \HR{} log-likelihood under positivity constraints is higher than the purely lasso-based methods, and we note that the best fit for the modified
$\ell_1$ penalty behaves better than the standard lasso method. 

We see that the number of edges remains quite constant for the
positivity constraint, except in the first data point. The first data point corresponds
to a value of $\lambda_n$ that is almost equal to $0$ (to guarantee the stability
of the algorithm under all constraints, we use only positive values). 
This estimate is different from the one that would be obtained with the original 
algorithm from \citet{piotrsteffen}. In their case there is no trace penalty and there is only one block of variables (instead of $2$ like in our case). This causes our algorithm
to behave in a more unstable way for smaller values of $\lambda_n$. For computing such problems 
without trace penalty and using only one block of variables, their paper
provides R code.

In \Cref{fig:maps} we can see the original flight graph (where edges denote direct flight connections), along with the maps estimated
when using our method with different Golazo penalties. For these maps, we pick values for the parameters that perform
optimally in validation.

\begin{figure}[htbp]
	\centering
	\begin{minipage}{0.45\textwidth}
	  \centering
	  \includegraphics[width=\linewidth]{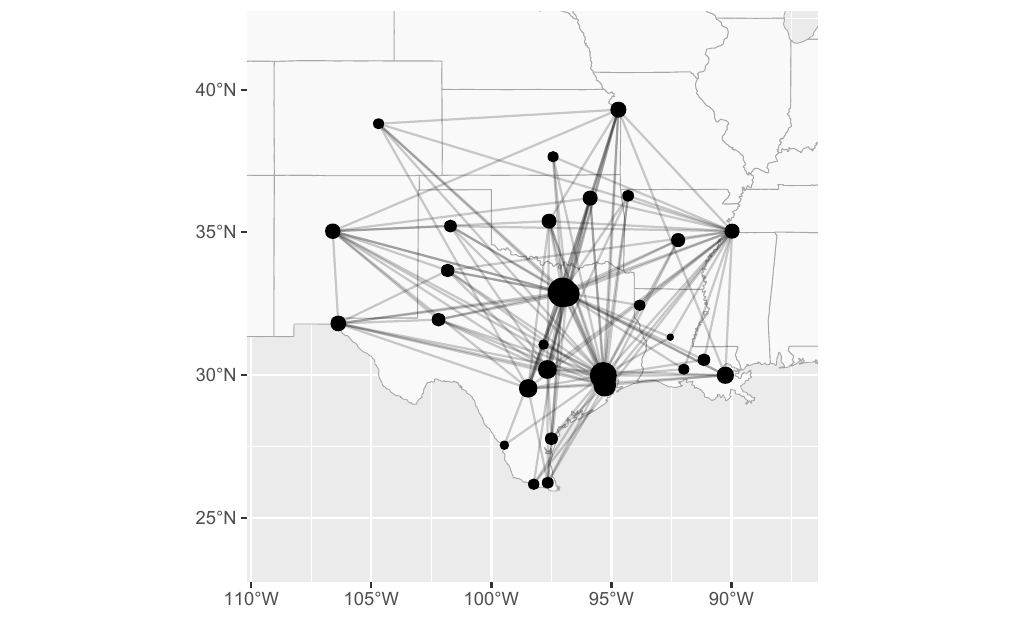}
	  \caption*{{(A)} Map with flight connections}
	\end{minipage}\hfill
	\begin{minipage}{0.45\textwidth}
	  \centering
	  \includegraphics[width=\linewidth]{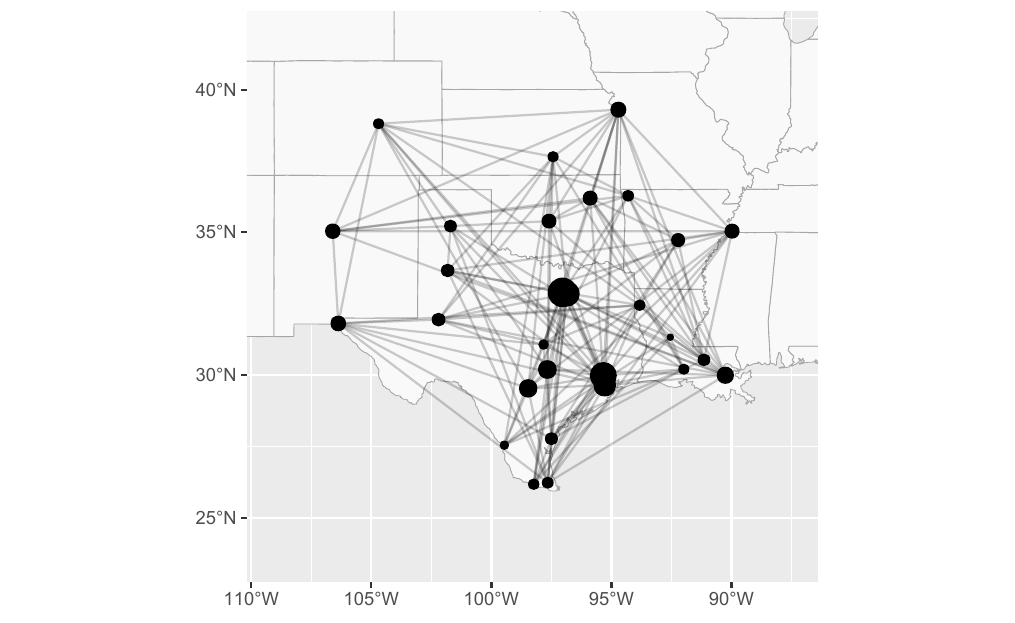}
	  \caption*{{(B)} Map for the standard $\ell_1$ penalty}
	\end{minipage}

	\vskip 0.5cm

	\begin{minipage}{0.45\textwidth}
	  \centering
	  \includegraphics[width=\linewidth]{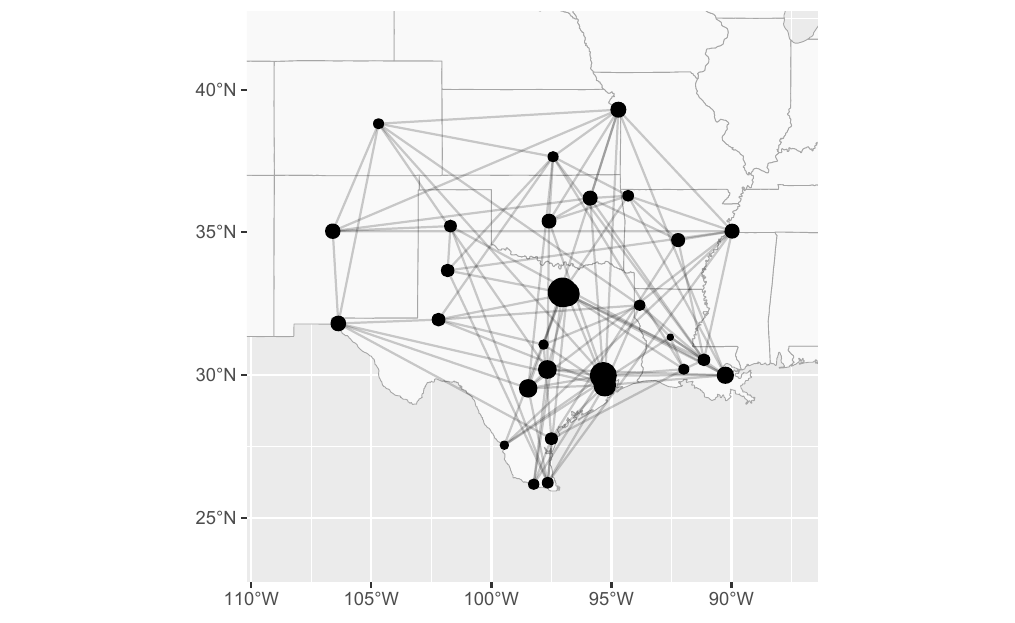}
	  \caption*{{(C)} Map for the $\ell_1$ penalty combined with the \EMTPtwo}
	\end{minipage}\hfill
	\begin{minipage}{0.45\textwidth}
		\centering
		\includegraphics[width=\linewidth]{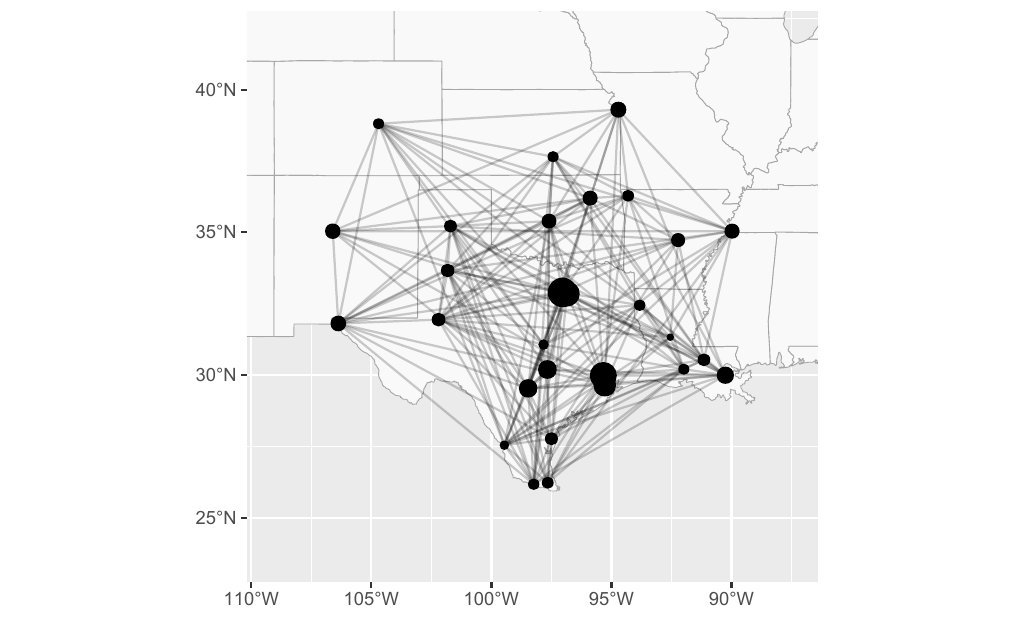}
		\caption*{{(D)} Map for \EMTPtwo constraint}
	  \end{minipage}
	\caption{Estimated maps for the different methods}\label{fig:maps}
  \end{figure}

\section{Discussion}
In this paper we propose generalized latent Gaussian graphical and \HR{} model learning via the Golazo penalty function.
We provide an ADMM algorithm that we apply to simulated and real data, and discuss various flexible penalization choices in comparison to the standard $\ell_1$-penalty.
In particular, the robustness of the \MTPtwo (and $ \text{EMTP}_2$) constraint  with respect to the hyperparameters provides an attractive alternative to settings when hyperparameter 
tuning is not possible (for instance, when training is too expensive). 
For future research, a main question beyond the scope of this paper is studying theoretical properties
and guarantees for the proposed methods, since the behavior depending on the Golazo penalty seems
like a promising research topic given what can be seen in experiments, where the results 
change a lot depending on the case.

Furthermore, one could explore whether some kind of ensemble of such estimators can improve 
performance over one estimator alone. This would be an
interesting practical improvement, since if a model is trained over multiple hyperparameters to obtain an optimal choice,
then suboptimal models could still be used as part of such an ensemble.
We would also like to consider in the future if performing a refit after a first fit of our model
can improve performance. This is something discussed in Appendix~I from \citet{engelke2024extremal}.
The idea behind this procedure is to, after obtaining a model with the procedure described in their paper,
fit a model optimizing the loglikelihood, using the previous estimate to constrain the sparsity
pattern and the column space of the hidden component. We believe it would be interesting to see
how useful this is depending on the setting, given that the variety of problems and constraints
suggests a large number of different possibilities.

\section{Acknowledgments}
We thank Sebastian Engelke and Armeen Taeb for their ideas and code, which have been very useful for the development of
this paper.

\bibliographystyle{chicago}
\bibliography{references}

\end{document}